\documentclass[a4paper,11pt]{amsart} 
\usepackage{amsmath,amsxtra,amssymb,latexsym, amscd,amsthm}
\usepackage[mathscr]{eucal}
\usepackage{mathrsfs}
\usepackage{bbm}
\usepackage{enumerate}
\usepackage{pict2e}
\usepackage{tikz}
\usepackage{graphicx}
\usepackage[a4paper]{geometry}
\usepackage{color}
\usepackage{hyperref}
\numberwithin{equation}{section}

\theoremstyle{plain}
\newtheorem{thm}{Theorem}[section]

\newtheorem{prp}[thm]{Proposition}

\newtheorem{lem}[thm]{Lemma}
\newtheorem*{ego*}{Egorov's Theorem}
\theoremstyle{definition}

\newtheorem{rem}[thm]{Remark}

\newtheorem{exa}[thm]{Example}
\newtheorem*{rem*}{Remark}

\newcommand{\dd}{\mathrm{d}}
\newcommand{\ee}{\mathrm{e}}
\newcommand{\ii}{\mathrm{i}}

\newcommand{\N}{\mathbb{N}}
\newcommand{\Z}{\mathbb{Z}}

\newcommand{\R}{\mathbb{R}}

\newcommand{\T}{\mathbb{T}}

\newcommand{\LL}{\mathbb{L}}

\newcommand{\cA}{\mathcal{A}}

\newcommand{\cH}{\mathcal{H}}

\newcommand{\fa}{\mathfrak{a}}
\newcommand{\fb}{\mathfrak{b}}

\DeclareSymbolFont{extraup}{U}{zavm}{m}{n}
\DeclareMathSymbol{\varheart}{\mathalpha}{extraup}{86}
\DeclareMathSymbol{\vardiamond}{\mathalpha}{extraup}{87}

\title[Ergodic theorems for quantum walks]{Ergodic theorems for continuous-time quantum walks on crystal lattices and the torus}
\author{Anne Boutet de Monvel, Mostafa Sabri}
\address{Institut de Math\'ematiques de Jussieu-Paris Rive Gauche, Universit\'e Paris-Cit\'e, 8 place Aur\'elie Nemours, case 7012, 75205 Paris Cedex 13, France.}
\email{anne.boutet-de-monvel@imj-prg.fr}
\address{Science Division, New York University Abu Dhabi, Saadiyat Island, Abu Dhabi, UAE.}
\email{mostafa.sabri@nyu.edu}
\usepackage{calc}
\usepackage{graphicx}

\makeatletter
\newlength{\temp@wc@width}
\newlength{\temp@wc@height}
\newcommand{\widecheck}[1]{%
  \setlength{\temp@wc@width}{\widthof{$#1$}}%
  \setlength{\temp@wc@height}{\heightof{$#1$}}%
  #1\hspace{-\temp@wc@width}%
  \raisebox{\temp@wc@height+2pt}[\heightof{$\widehat{#1}$}]%
     {\rotatebox[origin=c]{180}{\vbox to 0pt{\hbox{$\widehat{\hphantom{#1}}$}}}}%
}
\makeatother

\begin{document}

\begin{abstract}
We give several quantum dynamical analogs of the classical Kronecker-Weyl theorem, which says that the trajectory of free motion on the torus along almost every direction tends to equidistribute. As a quantum analog, we study the quantum walk $\exp(-\ii t \Delta) \psi$ starting from a localized initial state $\psi$. Then the flow will be ergodic if this evolved state becomes equidistributed as time goes on. We prove that this is indeed the case for evolutions on the flat torus, provided we start from a point mass, and we prove discrete analogs of this result for crystal lattices. On some periodic graphs, the mass spreads out non-uniformly, on others it stays localized. Finally, we give examples of quantum evolutions on the sphere which do not equidistribute.
\end{abstract}

\maketitle

\section{Introduction}

Classical ergodic theorems say that if $T$ is an ergodic transformation on some measure space $(\Omega,\mu)$, then averaging an observable $f$ over the trajectory under $T$ of a.e. point $x$ is the same as averaging the observable over the whole space:
\[
\lim_{n\to\infty}\frac{1}{n} \sum_{k=1}^n f(T^k x) = \frac{1}{\mu(\Omega)}\int_\Omega f(y)\,\dd\mu(y)
\]

Let us consider the case where the classical transformation is the geodesic flow. For the flat torus, the Kronecker-Weyl theorem says that for any $a\in C^0(\T^d)$ and $x_0\in\T^d$, if $y_0\in \R^d$ has rationally independent entries, then
\begin{equation}\label{e:KW}
\lim_{T\to\infty} \frac{1}{T}\int_0^T a(x_0+ty_0)\,\dd t = \int_{\T^d} a(x)\,\dd x\,.
\end{equation}

This means that the trajectory $\{x_0+ty_0\}_{t\ge 0}$ becomes uniformly distributed after large enough time, so that averaging a function over it is the same as the uniform average. 

In contrast, consider the standard Euclidean sphere $\mathbb{S}^2\subset \R^3$. This is a classical example in which the geodesic flow is not ergodic. A free particle moving with its kinetic energy simply travels along a great circle, its trajectory is very far from being dense in $\mathbb{S}^2$.

In this paper we are interested in giving quantum dynamical analogs of such results. Instead of starting from a point $x_0$ on the torus or sphere and integrating a test function over its trajectory $\phi_t(x_0)$, we will start from an initial state $\delta_{x_0}$ which is essentially a Dirac distribution at the point $x_0$, apply the evolution semigroup $\ee^{\ii t\Delta}\delta_{x_0}$ and check whether this state, which was highly localized at time zero, becomes equidistributed as time goes on. Our criterion for such an equidistribution is to compare $\int a(x) |(\ee^{\ii t\Delta}\delta_{x_0})(x)|^2\,\dd x$ with the uniform average $\int a(x)\,\dd x$ and show that they are close, for any test function $a(x)$. We will see that this is indeed the case for the flat torus, for the analogous discrete problem in $\Z^d$ and more generally for a large family of $\Z^d$-periodic lattices (which yield a more interesting mass profile), but untrue for the sphere.

The evolution semigroup $\ee^{-\ii tH}$ for a Hamiltonian $H$ is known as a \emph{continuous-time quantum walk} in the literature. The framework is usually to work on graphs such as $\Z^d$ in this context, but one can expect similarities in the continuum, which motivates the study of both cases in this paper. The terminology ``quantum walk'' is due to a quantum analogy with a random walk, which is more apparent in the case of \emph{discrete-time quantum walks}, the simplest example being a particle on $\Z$ walking (jumping a finite distance) under the action of a unitary operator $U$ at each time step $t=1,2,\dots$, and being in general in a superposition of states. We refer to \cite{ABNVW} for the basics and \cite{Por} for a systematic study. In contrast, in the continuous-time case, for $\ee^{-\ii t\cA_\Z}\delta_0$ we have a nonzero probability of being arbitrarily far from $0$ as soon as $t>0$, i.e. we can have $|(\ee^{-\ii t\cA_{\Z}}\delta_0)(n)|^2\neq 0$ with $n\gg 1$.

This paper is not the first work to give quantum analogs of ergodicity. This topic has first been explored by Shnirel'man, Colin de Verdi\`ere and Zelditch \cite{Shni,CdV,Zel87} and has since inspired research in many directions. The point of view of these \emph{quantum ergodicity} theorems is to show that in cases where the geodesic flow is ergodic, any orthonormal basis of eigenfunctions $(\psi_j)$ of the Laplace operator has a density one subsequence which becomes equidistributed in the high energy limit. More precisely, $|\psi_j(x)|^2\,\dd x$ approaches the uniform measure $\dd x$. Discrete analogs of this appeared for graphs. In this case, one considers instead a sequence of finite graphs $G_N$ converging in an appropriate sense (Benjamini-Schramm) to an infinite graph having a delocalized spectrum and shows an equidistribution property for the eigenfunctions of $G_N$, see e.g. \cite{ALM,A,AS,McKSa} and \cite{MTZ,NSSZ}. In the present work, our quantum interpretation is to follow instead how initially localized states (point masses) spread out under the action of the dynamics. This seems like a more direct translation of the classical picture. Our work in the continuum has relations with  \cite{Ma09,Maci,AR12,AM14,MaRi16} which we explain in \S~\ref{sec:ew}. A common difficulty in all the models we consider here is how to work with the high multiplicity of eigenvalues.

\subsection*{Notation}
As many articles in the quantum walks literature use the Dirac bra-ket notation, let us briefly explain the standard Hilbert space notation that we use in the paper. If $G$ is a graph, then $(\delta_v)_{v\in G}$ stands for the standard basis of $\ell^2(G)$ given by $\delta_v(u)=1$ if $u=v$ and zero otherwise. The scalar product is given by $\langle \phi,\psi\rangle = \sum_{v\in G} \overline{\phi(v)}\psi(v)$. Given a linear operator $A:\ell^2(G)\to \ell^2(G)$ we have in particular $\langle \phi, A\psi\rangle = \sum_{v\in G} \overline{\phi(v)}(A\psi)(v)$. If $a$ is a function $a(v)$, then $\langle \phi, a\phi \rangle = \sum_{v\in G} a(v) |\phi(v)|^2$.

In the bra-ket notation, $\{\delta_v\}_{v\in G}$ is replaced by $\{ |v\rangle: v\in G\}$. An operator $A$ acts on a vector $|\phi\rangle$ by $A|\phi\rangle$. The time evolution $\psi(t)=\ee^{-\ii tH} \psi$ is denoted by $|\psi(t)\rangle = U(t)|\psi\rangle$. Our $\langle \phi,\psi\rangle$ equals $\langle \phi|\psi\rangle$ and our $\langle \phi, A\psi\rangle$ equals $\langle \phi | A | \psi\rangle$. If $H$ is a Hamiltonian, expanding $\ee^{-\ii tH}\psi_0$ over an orthonormal eigenbasis $(\phi_j)$ of $H$ reads $\psi(t)=\ee^{-\ii tH} \psi_0 = \sum_j \langle \phi_j,\psi_0\rangle \ee^{-\ii t\lambda_j} \phi_j$ in our notation,  $|\psi(t)\rangle = \sum_j |\phi_j\rangle \ee^{-\ii t\lambda_j} \langle \phi_j|\psi_0\rangle$ in the bra-ket notation.

\medskip

We now discuss our results, first for the adjacency matrix on $\Z^d$, then more generally for periodic graphs in \S~\ref{sec:ingra}. We next move to the flat torus in \S~\ref{sec:intor}, then we conclude with the case of the sphere in \S~\ref{sec:insphe}.

\subsection{Case of graphs}\label{sec:ingra}
For transparency we begin with $\Z^d$. Consider a sequence of cubes $\Lambda_N = [\![0,N-1]\!]^d$, and let $A_N$ be the adjacency matrix on $\Lambda_N$ with periodic conditions. We denote the torus by $\T_\ast^d=[0,1)^d$. 

\begin{thm}\label{thm:zd}
For any $v\in \Lambda_N$, we have 
\begin{equation}\label{e:maindiszd}
\lim_{N\to\infty} \Big|\lim_{T\to\infty} \frac{1}{T}\int_0^T \langle \ee^{-\ii tA_N} \delta_v, a_N\ee^{-\ii tA_N} \delta_v\rangle\,\dd t - \langle a_N\rangle \Big| = 0\,,
\end{equation}
where $\langle a\rangle = \frac{1}{N^d}\sum_{u\in \Lambda_N} a(u)$, for the following class of observables $a_N$:
\begin{itemize}
\item $a_N(n) = f(n/N)$ for some $f\in H_s(\T_\ast^d)$ with $s>d/2$,
\item $a_N$ the restriction to $\Lambda_N$ of some $a\in \ell^1(\Z^d)$.
\end{itemize}
\end{thm}

Note that $\langle \ee^{-\ii tA_N} \delta_v, a\ee^{-\ii tA_N} \delta_v\rangle = \sum_{u\in\Lambda_N} a(u)|(\ee^{-\ii tA_N}\delta_v)(u)|^2$. So \eqref{e:maindiszd} shows that the probability density $\mu_{v,T}^N(u):=\frac{1}{T}\int_0^T |(\ee^{-\ii tA_N}\delta_v)(u)|^2\,\dd t$ on $\Lambda_N$ approaches the uniform density $\frac{1}{N^d}$, provided $T$ and $N$ are large enough, that is, the time average $\sum_{u\in\Lambda_N} a(u)\mu_{v,T}^N(u)$ approaches the space average $\frac{1}{N^d}\sum_{u\in\Lambda_N} a(u)$ in analogy to \eqref{e:KW}. See Figure~\ref{fig:sprea}.

A positive aspect of this result is that it holds \emph{for any $v$}, whereas in the eigenfunction interpretation, equidistribution only holds for a density one subsequence in general. The evolution moreover ``forgets the initial state $v$'', a known signature of ergodicity.

\begin{figure}[h!]
\begin{center}
\setlength{\unitlength}{0.6cm}
\thicklines
\begin{picture}(1,5)(-0.6,-2)
        \put(-9,0){\line(1,0){1.9}}
	\put(-9,-1){\line(1,0){1.9}}
	\put(-9,1){\line(1,0){1.9}}
	\put(-9,2){\line(1,0){1.9}}
    \put(-6.9,0){\line(1,0){1.8}}
    \put(-6.9,-1){\line(1,0){1.8}}
    \put(-6.9,1){\line(1,0){1.8}}
     \put(-6.9,2){\line(1,0){1.8}}
         \put(-4.9,0){\line(1,0){1.8}}
	\put(-4.9,-1){\line(1,0){1.8}}
	\put(-4.9,1){\line(1,0){1.8}}
	\put(-4.9,2){\line(1,0){1.8}}
    \put(-2.9,0){\line(1,0){1.8}}
	\put(-2.9,-1){\line(1,0){1.8}}
	\put(-2.9,1){\line(1,0){1.8}}
	\put(-2.9,2){\line(1,0){1.8}}
\put(-7,-2){\line(0,1){0.9}}
	 \put(-5,-2){\line(0,1){0.9}}
	 \put(-3,-2){\line(0,1){0.9}}
\put(-7,-0.9){\line(0,1){0.8}}
	 \put(-5,-0.9){\line(0,1){0.8}}
	 \put(-3,-0.9){\line(0,1){0.8}}
\put(-7,0.1){\line(0,1){0.8}}
	 \put(-5,0.1){\line(0,1){0.8}}
	 \put(-3,0.1){\line(0,1){0.8}}
\put(-7,1.1){\line(0,1){0.8}}
	 \put(-5,1.1){\line(0,1){0.8}}
	 \put(-3,1.1){\line(0,1){0.8}}
\put(-7,2.1){\line(0,1){0.8}}
	 \put(-5,2.1){\line(0,1){0.8}}
	 \put(-3,2.1){\line(0,1){0.8}}
	 \put(-5,-1){\circle{.2}}
	 \put(-5,0){\circle{.2}}
	 \put(-7,-1){\circle{.2}}
	 \put(-7,0){\circle{.2}}
	 \put(-3,-1){\circle{.2}}
	 \put(-3,0){\circle{.2}}
	 \put(-5,1){\textcolor{blue}{\circle*{.25}}}
	 \put(-5,2){\circle{.2}}
	 \put(-7,1){\circle{.2}}
	 \put(-7,2){\circle{.2}}
	 \put(-3,1){\circle{.2}}
	 \put(-3,2){\circle{.2}}
	 \put(1,0){\line(1,0){8}}
	 \put(1,-1){\line(1,0){8}}
	\put(1,1){\line(1,0){8}}
	\put(1,2){\line(1,0){8}}
	 \put(3,-2){\line(0,1){5}}
	 \put(5,-2){\line(0,1){5}}
	 \put(7,-2){\line(0,1){5}}
	 \put(5,-1){\textcolor{cyan}{\circle*{.22}}}
	 \put(5,0){\textcolor{cyan}{\circle*{.22}}}
	 \put(3,-1){\textcolor{cyan}{\circle*{.22}}}
	 \put(3,0){\textcolor{cyan}{\circle*{.22}}}
	 \put(7,-1){\textcolor{cyan}{\circle*{.22}}}
	 \put(7,0){\textcolor{cyan}{\circle*{.22}}}
	 \put(5,1){\textcolor{cyan}{\circle*{.22}}}
	 \put(5,2){\textcolor{cyan}{\circle*{.22}}}
	 \put(3,1){\textcolor{cyan}{\circle*{.22}}}
	 \put(3,2){\textcolor{cyan}{\circle*{.22}}}
	 \put(7,1){\textcolor{cyan}{\circle*{.22}}}
	 \put(7,2){\textcolor{cyan}{\circle*{.22}}}
\end{picture}
\caption{Left: Point mass $\delta_v$ at time zero. Right: The density $\mu_{v,T}^N$ for $T,N\gg 0$. The point mass spreads out uniformly, a very strong form of delocalization.}\label{fig:sprea}
\end{center}
\end{figure}
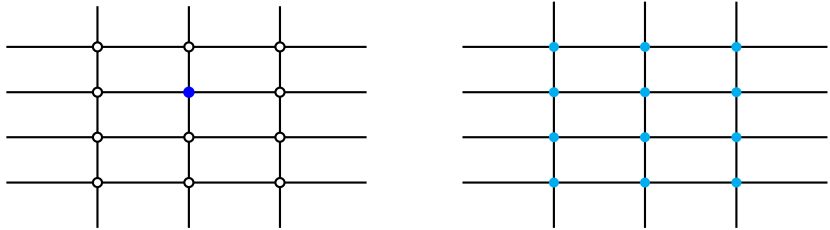 

The first class of observables allows taking bump functions $f$ supported on balls $B_R(x_0)\subset \T_\ast^d$. Then the result implies that $\sum_{u\in\Lambda_N} f(u/N)\mu_{v,T}^N(u) \approx \frac{1}{N^d} \sum_{u\in \Lambda_N} f(n/N) \approx \int_{\T_\ast^d} f(x)\,\dd x$, which is independent of $x_0$. This implies that for any macroscopic ball $B\subset \Lambda_N$ of size $|B|=\alpha N^d$, we have $\mu_{v,T}^{(N)}(B) \approx \alpha$ for $N,T\gg 0$ (e.g. take $\alpha=\frac{1}{2}$ or $\alpha=\frac{1}{4}$ and vary $B$).

We may extend Theorem~\ref{thm:zd} as follows:

\begin{thm}\label{thm:zdexten}
Under the same assumptions on $a_N$, we also have for $v\neq w$,
\begin{equation}\label{e:delvw}
\lim_{N\to\infty}\lim_{T\to\infty} \frac{1}{T}\int_0^T \langle \ee^{-\ii tA_N}\delta_v, a_N\ee^{-\ii tA_N}\delta_w\rangle\,\dd t = 0\,.
\end{equation}
More generally, for any $\phi,\psi$ of compact support, we have
\begin{equation}\label{e:ephipsi}
\lim_{N\to\infty}\Big|\lim_{T\to\infty} \frac{1}{T}\int_0^T \langle \ee^{-\ii tA_N}\phi, a_N\ee^{-\ii tA_N}\psi\rangle\,\dd t - \langle a_N\rangle\langle\phi,\psi\rangle\Big| = 0\,.
\end{equation}
\end{thm}

There are many natural questions that arise when looking at these results.

First, why not work on $\Z^d$ directly? One issue is that the dynamics are dispersive \cite{SK} on $\Z^d$, more precisely $|\ee^{-\ii t \cA_{\Z^d}}\delta_v(w)|^2  = \prod_{j=1}^d |J_{v_j-w_j}(2t)|^2$, where $J_k$ is the Bessel function of order $k$, and $|J_k(t)|\lesssim t^{-1/3}$ uniformly in $k$, see \cite{Lan}. So as time grows large, the probability measure $|\ee^{-\ii t \cA_{\Z^d}}\delta_v(w)|^2$ on $\Z^d$ simply converges to zero. One could instead consider the limit of the process rescaled per unit time. We did this previously in \cite{BS} and computed the limiting measure explicitly.

Still, this does not answer whether we can possibly invert the order of limits in the theorem. The answer is in fact negative.\footnote{Still, we can take the two limits simultaneously to infinity in a certain regime. See Remark~\ref{rem:t=n}.}

\begin{prp}\label{prp:invertlimzd}
There exists a sequence of observables $a_N$ on $\Lambda_N$ of the form $f(n/N)$ such that
\[
\liminf_{N\to \infty} \left( \langle \ee^{-\ii tA_N} \delta_v, a_N \ee^{-\ii tA_N}\delta_v\rangle - \langle a_N\rangle   \right) \ge \frac{1}{2}
\]
for all time. The same statement holds for the averaged dynamics $\frac{1}{T}\int_0^T$.
\end{prp}

This indicates that the limit over time should be considered first. But can one get rid of the time average and consider the limit directly in $t$ ? The answer is negative.

\begin{prp}\label{prp:noavzd}
There exists a sequence of observables $a_N$ on $\Lambda_N$ of the form $f(n/N)$ such that $\langle \ee^{-\ii tA_N} \delta_v, a_N \ee^{-\ii tA_N}\delta_v\rangle$ has no limit as $t\to\infty$.

The limit of the corresponding average $\int_{T-1}^T\langle \ee^{-\ii tA_N} \delta_v, a_N \ee^{-\ii tA_N}\delta_v\rangle\,\dd t$ also doesn't exist.
\end{prp}

The last point illustrates that a ``full'' average $\frac{1}{T}\int_0^T$ is needed.

Proposition~\ref{prp:invertlimzd} shows that we should consider the large time limit first. But is it actually necessary to take $N$ to infinity? The answer is yes.

\begin{prp}\label{prp:ves}
For each $N$ there exists $a_N$ such that
\[
\lim_{T\to\infty} \frac{1}{T}\int_0^T \langle \ee^{-\ii tA_N}\delta_v, a_N \ee^{-\ii tA_N}\delta_v\rangle - \langle a_N\rangle = b_N(v)
\]
with $b_N(v)\neq 0$. Here, $b_N(v)\to 0$ as $N\to\infty$ at the rate $N^{-1}$.
\end{prp}

Proposition~\ref{prp:ves} shows that we cannot expect a faster rate of convergence in Theorem~\ref{thm:zd} than $N^{-1}$. We indeed achieve this upper bound in the proof in Section~\ref{sec:pro}.

\medskip

Theorem~\ref{thm:zd} can be generalized to $\Z^d$-periodic graphs (crystals). This requires some vocabulary which we prefer to postpone to Section~\ref{sec:pergra}, so we will explain the theorem in words here instead and refer to Theorem~\ref{thm:pergra} for a more precise statement.

\begin{thm}
Theorem~\ref{thm:zd} holds true more generally for periodic Schr\"odinger operators on $\Z^d$-periodic graphs, provided they satisfy a certain Floquet condition. This condition is satisfied in particular for the adjacency matrix on infinite strips, on the honeycomb lattice, and for Sch\"odinger operators with periodic potentials on the triangular lattice and on $\Z^d$, for any $d$. The average $\langle a_N\rangle$ may not be the uniform average of $a_N$ in general, but a certain weighted average.
\end{thm}

See Figures~\ref{fig:strips} and \ref{fig:carte} for examples of this.

\begin{figure}[h!]
\begin{center}
\setlength{\unitlength}{0.9cm}
\thicklines
\begin{picture}(1.3,1.9)(-0.8,-2.1)
   \put(-8,-0.5){\line(1,0){4}}
	 \put(-8,-1.5){\line(1,0){4}}
	 \put(-7.5,-1.5){\line(0,1){1}}
	 \put(-6.5,-1.5){\line(0,1){1}}
	 \put(-5.5,-1.5){\line(0,1){1}}
	 \put(-4.5,-1.5){\line(0,1){1}}
	 \put(-7.5,-1.5){\textcolor{gray}{\circle*{.2}}}
	 \put(-7.5,-0.5){\textcolor{gray}{\circle*{.2}}}
	 \put(-6.5,-1.5){\textcolor{gray}{\circle*{.2}}}
	 \put(-6.5,-0.5){\textcolor{gray}{\circle*{.2}}}
	 \put(-5.5,-1.5){\textcolor{gray}{\circle*{.2}}}
	 \put(-5.5,-0.5){\textcolor{gray}{\circle*{.2}}}
	 \put(-4.5,-1.5){\textcolor{gray}{\circle*{.2}}}
	 \put(-4.5,-0.5){\textcolor{gray}{\circle*{.2}}}
	  \put(-2.5,0){\line(1,0){4}}
	 \put(-2.5,-1){\line(1,0){4}}
	 \put(-2.5,-2){\line(1,0){4}}
	 \put(-2,-1){\line(0,1){1}}
	 \put(-1,-1){\line(0,1){1}}
	 \put(0,-1){\line(0,1){1}}
	 \put(1,-1){\line(0,1){1}}
	 \put(-2,-2){\line(0,1){1}}
	 \put(-1,-2){\line(0,1){1}}
	 \put(0,-2){\line(0,1){1}}
	 \put(1,-2){\line(0,1){1}}
	 \put(-2,-1){\textcolor{pink}{\circle*{.2}}}
	 \put(-2,0){\textcolor{red}{\circle*{.2}}}
	 \put(-2,-2){\textcolor{red}{\circle*{.2}}}
	 \put(-1,-1){\textcolor{pink}{\circle*{.2}}}
	 \put(-1,0){\textcolor{red}{\circle*{.2}}}
	 \put(-1,-2){\textcolor{red}{\circle*{.2}}}
	 \put(0,-1){\textcolor{pink}{\circle*{.2}}}
	 \put(0,0){\textcolor{red}{\circle*{.2}}}
	 \put(0,-2){\textcolor{red}{\circle*{.2}}}
	 \put(1,-1){\textcolor{pink}{\circle*{.2}}}
	 \put(1,0){\textcolor{red}{\circle*{.2}}}
	 \put(1,-2){\textcolor{red}{\circle*{.2}}}
	  \put(3,0){\line(1,0){4}}
	 \put(3,-1){\line(1,0){4}}
	 \put(3,-2){\line(1,0){4}}
	 \put(3.5,-1){\line(0,1){1}}
	 \put(4.5,-1){\line(0,1){1}}
	 \put(5.5,-1){\line(0,1){1}}
	 \put(6.5,-1){\line(0,1){1}}
	 \put(3.5,-2){\line(0,1){1}}
	 \put(4.5,-2){\line(0,1){1}}
	 \put(5.5,-2){\line(0,1){1}}
	 \put(6.5,-2){\line(0,1){1}}
	 \put(3.5,-1){\textcolor{blue}{\circle*{.2}}}
	 \put(3.5,0){\textcolor{cyan}{\circle*{.2}}}
	 \put(3.5,-2){\textcolor{cyan}{\circle*{.2}}}
	 \put(4.5,-1){\textcolor{blue}{\circle*{.2}}}
	 \put(4.5,0){\textcolor{cyan}{\circle*{.2}}}
	 \put(4.5,-2){\textcolor{cyan}{\circle*{.2}}}
	 \put(5.5,-1){\textcolor{blue}{\circle*{.2}}}
	 \put(5.5,0){\textcolor{cyan}{\circle*{.2}}}
	 \put(5.5,-2){\textcolor{cyan}{\circle*{.2}}}
	 \put(6.5,-1){\textcolor{blue}{\circle*{.2}}}
	 \put(6.5,0){\textcolor{cyan}{\circle*{.2}}}
	 \put(6.5,-2){\textcolor{cyan}{\circle*{.2}}}
	 \end{picture}
\caption{A point mass on the ladder (left) eventually equidistributes. On the strip of width $3$, if the initial point mass lies in the top layer, it eventually (center) puts $\frac{3}{8}$ of its mass over the top and bottom layers and $\frac{1}{4}$ on the middle. If the point mass was in the middle layer, it eventually (right) puts $\frac{1}{4}$ on the top and bottom layers and $\frac{1}{2}$ in the middle layer.}\label{fig:strips}
\end{center}
\end{figure}
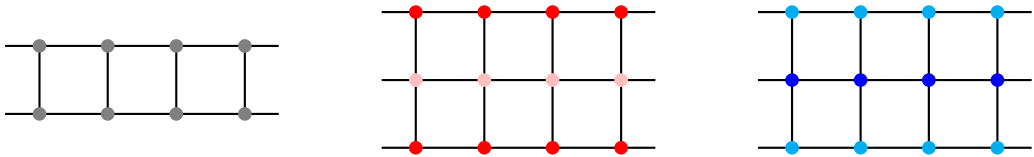

In this result we focus on the two most natural choices of initial states: point masses (Theorem~\ref{thm:pergra}) and initial states uniformly spread over a single fundamental cell (Remark~\ref{rem:perini}). Still, we give an expression \eqref{e:gena} for more general states.

As we explain in the end of Section~\ref{sec:pergra}, some Floquet condition must be assumed at least to rule out flat bands.

\subsection{Torus dynamics}\label{sec:intor}
We now turn our attention to the continuum and consider the torus $\T_\ast^d$. This model is rather unusual compared to the quantum walk literature, which typically studies evolutions on graphs, but has the advantage of being directly comparable to the Kronecker-Weyl theorem \eqref{e:KW}. The technical issue however is that we can no longer consider Dirac distributions $\delta_x$ directly as in the case of graphs. We will thus regularize (approximate) them in two ways: in momentum space, then in position space.

For the momentum-space approximation, we consider $\delta_y^E:=\frac{1}{\sqrt{N_E}} \mathbf{1}_{(-\infty,E]}(-\Delta)\delta_y$, where $N_E$ is the number of Laplacian eigenvalues in $(-\infty,E]$. This truncated Dirac function has previously been considered in \cite{AR12}. In our case, $\delta_y^E$ is a trigonometric polynomial, see \S~\ref{sec:conmo}.
\begin{thm}\label{thm:ctstor}
We have for any $T>0$,
\begin{enumerate}[\rm(1)]
\item For any $y\in \T_\ast^d$, any $a\in H^s(\T_\ast^d)$, $s>d/2$,
\[
\lim_{E\to\infty} \frac{1}{T}\int_0^T \langle \ee^{\ii t\Delta} \delta_y^E, a \ee^{\ii t\Delta} \delta_y^E\rangle\,\dd t = \int_{\T_\ast^d} a(x)\,\dd x\,.
\]
\item If $x\neq y$, then for $a\in H^s(\T_\ast^d)$, $s>d/2$,
\[
\lim_{E\to\infty}  \frac{1}{T}\int_0^T \langle \ee^{\ii t\Delta} \delta_x^E, a \ee^{\ii t\Delta} \delta_y^E\rangle\,\dd t = 0\,.
\]
\item Result \emph{(1)} remains true if $a=a^E$ depends on the semiclassical parameter $E$, as long as all partial derivatives of order $\le s$ are uniformly bounded by $c E^{r}$ for some $r<\frac{1}{4}$. More precisely, $\lim\limits_{E\to\infty} |\frac{1}{T}\int_0^T \langle \ee^{\ii t\Delta}\delta_y^E, a^E\ee^{\ii t\Delta}\delta_y^E\rangle\,\dd t - \int_{\T_\ast^d} a^E(x)\,\dd x| =0$.
\item The probability measure $\dd\mu_{y,T}^E(x)=(\frac{1}{T}\int_0^T |\ee^{\ii t\Delta} \delta_y^E(x)|^2\,\dd t)\dd x$ on $\T_\ast^d$ converges weakly to the uniform measure $\dd x$ as $E\to\infty$.
\end{enumerate}
\end{thm}
Hence, averaging $a$ over $\dd\mu_{y,T}^E(x)=(\frac{1}{T}\int_0^T|\ee^{\ii t\Delta} \delta_y^E(x)|^2\dd t)\,\dd x$ is the same as averaging $a$ over the uniform measure $\dd x$, after the initial state becomes sufficiently localized.

Remarkably, equidistribution occurs immediately if we are initially sufficiently close to a Dirac distribution. We do not need to wait for large time $T$. Compare with \eqref{e:KW}.\footnote{One interpretation is that in semiclassical analysis, one considers evolutions of $-h^2\Delta$ instead, with $h\to 0$, so in this sense, we do study long times.}

The third point allows taking observables of shrinking support. This problem was recently studied by \cite{Han,HR}. More precisely, we can allow observables concentrated near any $x_0\in \T_\ast^d$, with support shrinking like $E^{-\beta}$ for $\beta = \frac{1}{2(d+1)}$. In fact, starting from any fixed smooth $a$ supported in a ball $B_R$ around the origin, define $a^{E,x_0}(x) = a(E^\beta(x-x_0))$. This is supported in the ball $B_{R E^{-\beta}}(x_0)$. It satisfies $\partial_{x_i}^k a^{E,x_0} = E^{\beta k} \partial_{x_i}^k a(E^{\beta}(x-x_0))$, hence $\|\partial_{x_i}^k a^{E,x_0}\|_\infty \le E^{\beta k} \|\partial_{x_i}^ka\|_\infty$. We thus need $\beta s<\frac{1}{4}$ for point (3). We also need $s>d/2$ to respect the assumptions. Choosing $s=\frac{d}{2}+\frac{1}{4}$, we see that $\beta = \frac{1}{2(d+1)}$ suffices.

This gives an even more precise result of the intuition that the mass of $|\ee^{\ii t\Delta} \delta_y^E(x)|^2\,\dd x$ equidistributes on average. It says that this equidistribution property remains true if we zoom in near any point $x_0\in \T_\ast^d$.

Interestingly, $(\ee^{\ii t \Delta} \delta_y^E)(x) = \frac{1}{\sqrt{N_E}}\sum_{\lambda_\ell \le E} \overline{e_\ell(y)} \ee^{-\ii t\lambda_\ell} e_\ell(x)$ with $\lambda_\ell=4 \pi^2 \ell^2$ is a normalized \emph{truncated theta function} if $d=1$, which is important in number theory \cite{Hia,Kuz}.

Theorem~\ref{thm:ctstor} extends immediately to more general energy cutoffs of $\delta_y$, i.e. there is some flexibility in the choice of $\delta_y^E$. We can also consider more general tori $\T=\times_{i=1}^d[0,b_i)$. See \S~\ref{sec:ctsgen} for details. As before,

\begin{lem}\label{lem:tav}
Time averaging is necessary, even when $\lim_{E\to\infty} \langle \ee^{\ii t\Delta}\delta_y^E,a\ee^{\ii t\Delta}\delta_y^E\rangle$ exists, it generally depends on the value of $t$ and it may not be equal to $\int a(x)\,\dd x$.
\end{lem}

We next give a result by approximating in position space.

\begin{thm}\label{thm:cont1}
Fix any $y\in \T_\ast^d$ and consider $\phi_y^\varepsilon = \frac{1}{\sqrt{\varepsilon_1\varepsilon_2\cdots\varepsilon_d}} \mathbf{1}_{\times_{i=1}^d [y_i,y_i+\varepsilon_i]}$. Then for any $a\in H^s(\T_\ast^d)$, $s>d/2$, $T>0$,
\[
\lim_{\varepsilon\downarrow 0} \frac{1}{T}\int_0^T \langle \ee^{\ii t\Delta} \phi_y^\varepsilon, a \ee^{\ii t\Delta} \phi_y^\varepsilon\rangle\,\dd t = \int_{\T_\ast^d} a(x)\,\dd x\,,
\]
where $\varepsilon\downarrow 0$ means more precisely that $\varepsilon_i\downarrow 0$ for each $i$. If $x\neq y$, then
\[
\lim_{\varepsilon\downarrow 0} \frac{1}{T}\int_0^T \langle \ee^{\ii t\Delta} \phi_x^\varepsilon, a \ee^{\ii t\Delta} \phi_y^\varepsilon\rangle\,\dd t = 0\,.
\]

Furthermore, the probability measure $\dd\mu_{y,T}^\varepsilon(x)=(\frac{1}{T}\int_0^T |\ee^{\ii t\Delta} \phi_y^\varepsilon(x)|^2\,\dd t)\dd x$ on $\T_\ast^d$ converges weakly to the uniform measure $\dd x$ as $\varepsilon\downarrow 0$.
\end{thm}

This theorem says that once our initial state is close enough (in position space) to a Dirac mass, then its averaged dynamics will become equidistributed.

Theorem~\ref{thm:cont1} is actually valid for a more general class of initial states $(\phi_\varepsilon)$, see \S~\ref{sec:conpo}.

Theorems~\ref{thm:ctstor} and \ref{thm:cont1} remain true if we first take the limit over $T$ and then over $E$, i.e.
\begin{equation}\label{e:oldcont}
\lim_{E\to\infty}\lim_{T\to\infty} \frac{1}{T}\int_0^T \langle \ee^{\ii t\Delta} \delta_y^E, a \ee^{\ii t\Delta} \delta_y^E\rangle\,\dd t = \int_{\T_\ast^d} a(x)\,\dd x\,.
\end{equation}
This is in fact easier to prove, see Remark~\ref{rem:cts2lim}. What is remarkable is that we don't need to consider a large time in Theorems~\ref{thm:ctstor} and \ref{thm:cont1} for equidistribution to occur. Also notice that the limits over $E$ and $T$ can be interchanged, in view of the theorems and \eqref{e:oldcont}.

\subsection{Sphere dynamics}\label{sec:insphe}
As we mentioned in the introduction, classical evolution on the sphere is far from ergodic. We have a confirmation of this as follows.

\begin{thm}\label{thm:sphe}
Fix $\xi\in\mathbb{S}^{d-1}$. There exists a normalized approximate Dirac distribution $S_\xi^{(n)}$ such that $\lim_{T\to\infty} \frac{1}{T}\int_0^T |\ee^{\ii t\Delta} S_\xi^{(n)}(\eta)|^2\,\dd t$ is not equidistributed on the sphere as $n\to\infty$ and actually diverges for $\eta=\xi$.

If $d=3$, there exists an observable $a$ such that the analog of \eqref{e:oldcont} is violated.
\end{thm}

\subsection{Earlier results and perspectives}\label{sec:ew}
The time evolution of quantum walks is a central topic. Let us mention \cite{Kar,CLR} in relation to mixing time. We are not aware of earlier works showing equidistribution, and it would be very interesting to see which discrete-time quantum walks satisfy this phenomenon. We mention \cite{MSA} in which the evolution of a Grover walk in large boxes in $\Z^2$ was shown to localize, and \cite{DD} where (non)-thermalization of fullerene graphs was investigated. See also \cite{ST} for thermalization in a free fermion chain.

A study of the quantum dynamics on the torus appeared previously in the more general setting of Schr\"odinger operators in \cite{AM14,Maci}. It is shown in \cite{AM14} that if $(u_n)$ is a sequence in $L^2(\T_\ast^d)$ such that $\|u_n\| = 1$, if $H = -\Delta+ V$ is a Schr\"odinger operator on $\T_\ast^d$ and if $\dd \mu_n(x) = (\int_0^1 |(\ee^{-\ii tH} u_n)(x)|^2\,\dd t)\,\dd x$, then any weak limit of $\mu_n$ is absolutely continuous.

This is much broader than our framework. As a special case, this result implies that any weak limit of the measures $\mu_{y,T}^E$ and $\mu_{y,T}^\varepsilon$ in Theorems~\ref{thm:ctstor} and \ref{thm:cont1} is absolutely continuous. More general statements regarding Wigner distributions can also be found in \cite{AM14}.

At this level of generality, this result cannot be improved to ensuring convergence to the uniform measure. For example, if we take $d=1$, $V=0$ and $u_n(x) = \sqrt{2}\cos(2\pi x)$ for all $n$, then $\dd \mu_n(x) = (\int_0^1 2|\ee^{-\ii t (4\pi^2)} \cos(2\pi x)|^2\,\dd t)\,\dd x = 2\cos^2(2\pi x)\,\dd x$, which is not the uniform measure. Similarly, if we take $u_n(x) = \sqrt{2}\cos(2\pi x)$ for even $n$ and $u_n(x) = \sqrt{2}\sin(2\pi x)$ for odd $n$, we see that $\mu_n$ does not converge, having two limit points.


As this preprint was being circulated, Maxime Ingremeau and Fabricio Maci\`a explained to us that it is possible to prove the first point in each of Theorem~\ref{thm:ctstor} and \ref{thm:cont1} by first computing the semiclassical measures of $(\delta_y^E)$ and $(\phi_y^\varepsilon)$, which lives in phase space $T^\ast \T_\ast^d$. There should be a unique limit for each sequence, of the form $\mu_0(\dd x,\dd \xi)=\delta_y f(\xi)\,\dd \xi$, which does not charge the resonant frequencies, so one could apply \cite[(8) and Prp. 1]{Maci}, which rely on the microlocal analysis developed in \cite{Ma09}. It seems this can work even in the presence of a potential if one uses \cite[Th. 3]{AM14}. Our proof on the other hand is very simple, using explicit computations, we make no use of microlocal analysis as our framework is more special, because we had different aims in mind. Concerning the sphere, one could consider sequences of the form $\rho_y^h(x)=h^{-d/2}\rho(\frac{x-y}{h})$, where $\rho$ is an $L^2$ coordinate chart. The semiclassical measure will have a similar form. Using \cite[Thm. 4]{Ma09}, one can then deduce that $\frac{1}{T}\int_0^T |\ee^{\ii t\Delta}\rho_y^h(x)|^2\,\dd t$ will be absolutely continuous as $h\to 0$ (more precisely, a weighted superposition of uniform orbit measures modulated by $|\widehat{\rho}(\xi)|^2$). In the same spirit, one can work a bit and use \cite[Prp. 2.2.(i) and Th.4.3]{MaRi16} to deduce that $\frac{1}{T}\int_0^T |\ee^{\ii t\Delta}S_\xi^{(n)}(\eta)|^2\,\dd t$ is absolutely continuous as $n\to \infty$. This is not exactly our aim in Theorem~\ref{thm:sphe}, but is an interesting complementary information.


As for negative curvature, the authors in \cite{AR12} consider the case of a compact Anosov manifold $M$ and show that if $\delta_y^h$ is an $h$-truncated Dirac distribution, then as $h\to 0$, $\frac{1}{T}\int_0^T \langle \ee^{\ii t\Delta/2} \delta_y^h, \mathrm{Op}_h(a) \ee^{\ii t\Delta/2}\delta_y^h\rangle\,\dd t \approx \int_{S^\ast M} a\,\dd L$ \emph{for most $y$}. This is another instance where equidistribution occurs immediately once $\delta_y^h$ becomes close enough to $\delta_y$: there is no need to take $T\to\infty$. The fact that it holds for most $y$ means more precisely that the volume of $y\in M$ where this doesn't hold vanishes as $h\to0$. In this respect, the fact that our equidistribution results for the torus (and graphs) hold \emph{for each $y$} is worth emphasizing.

We finally mention the paper \cite{Schu} in the context of the evolution of Lagrangian states. These are localized in speed rather than in position.

As we mentioned earlier, there is a large literature on eigenfunction quantum ergodicity, in particular \cite{MR,Zel}. It is natural to ask if this property is related to the present quantum dynamical picture. We discuss this in Appendix~\ref{app:a}. In particular, while there are several proofs of eigenfunction ergodicity for regular graphs with few cycles \cite{A}, it is not very clear how to prove the dynamical criterion in that context; this seems like an interesting direction for future considerations.

\section{Case of the integer lattice}\label{sec:pro}

Here we prove Theorem~\ref{thm:zd} and \ref{thm:zdexten} and Propositions~\ref{prp:invertlimzd}--\ref{prp:ves}. Throughout, $a:=a_N$.

\begin{proof}[Proof of Theorem~\ref{thm:zd}]
Consider the orthonormal basis $e_m^{(N)}(n) = \frac{1}{N^{d/2}}\ee^{2\pi\ii m\cdot n/N}$. Given $\psi = \sum_{\ell\in \Lambda_N} \psi_\ell^{(N)} e_\ell^{(N)}$, where $\psi_\ell^{(N)}=\langle e_\ell^{(N)},\psi\rangle$, since $A_N e_k^{(N)} = \lambda_k^{(N)} e_k^{(N)}$, for $\lambda_k^{(N)} = \sum_{i=1}^d 2\cos(\frac{2\pi k_i}{N})$, we have $\ee^{-\ii tA_N}\psi = \sum_{\ell\in \Lambda_N} \psi_\ell^{(N)} \ee^{-\ii t\lambda_\ell^{(N)}} e_\ell^{(N)}$. 

Expand $a = \sum_{m \in \Lambda_N} a_m^{(N)} e_m^{(N)}$. Then using that $e_m^{(N)}e_\ell^{(N)} = \frac{1}{N^{d/2}} e_{\ell+m}^{(N)}$, we obtain
\[
(\ee^{\ii tA_N} a\ee^{-\ii tA_N} \psi)(n) = \sum_{\ell,m\in\Lambda_N} \psi_\ell^{(N)} \ee^{\ii t(\lambda_{\ell+m}^{(N)}-\lambda_\ell^{(N)})} a_m^{(N)} e_m^{(N)}(n) e_\ell^{(N)}(n)\,.
\]

In particular, as $\psi_\ell^{(N)} = \overline{e_\ell^{(N)}(v)}$ for $\psi = \delta_v$, we get
\begin{align}\label{e:scalex}
\langle \ee^{-\ii tA_N} \delta_v, a\ee^{-\ii tA_N} \delta_v\rangle &= (\ee^{\ii tA_N} a\ee^{-\ii tA_N} \delta_v)(v) \nonumber\\
&=\frac{1}{N^d}\sum_{\ell,m\in\Lambda_N} \ee^{\ii t(\lambda_{\ell+m}^{(N)}-\lambda_\ell^{(N)})} a_m^{(N)} e_m^{(N)}(v)\nonumber\\
&= \frac{1}{N^d} \sum_{w\in \Lambda_N} a(w) + \frac{1}{N^d} \sum_{\substack{m\in\Lambda_N\\ m\neq 0}}  a_m^{(N)} e_m^{(N)}(v)\sum_{\ell\in \Lambda_N}  \ee^{\ii t(\lambda_{\ell+m}^{(N)}-\lambda_\ell^{(N)})} \,,
\end{align}
where we used that $a_0^{(N)}e_0^{(N)}(v) = \frac{1}{N^d}\sum_{w\in \Lambda_N} a(w)$ for any $v$. 

If $\lambda_{\ell+m}^{(N)}\neq \lambda_\ell^{(N)}$ then $\frac{1}{T}\int_0^T \ee^{\ii t(\lambda_{\ell+m}^{(N)}-\lambda_\ell^{(N)})}\,\dd t = \frac{\ee^{\ii T(\lambda_{\ell+m}^{(N)}-\lambda_{\ell}^{(N)})}-1}{T(\lambda_{\ell+m}^{(N)}-\lambda_\ell^{(N)})}\to 0$ as $T\to\infty$. Thus,
\begin{multline}\label{e:limdis}
\lim_{T\to\infty} \frac{1}{T} \int_0^T \langle \ee^{-\ii tA_N}\delta_v,a\ee^{-\ii tA_N}\delta_v\rangle - \langle a\rangle \\
= \frac{1}{N^d}\sum_{\substack{m\in\Lambda_N\\ m\neq 0}} a_m^{(N)} e_m^{(N)}(v)\cdot \#\{ \ell\in \Lambda_N:\lambda_{\ell+m}^{(N)}=\lambda_\ell^{(N)}\} \,.
\end{multline}

Let $A_m = \{ \ell\in \Lambda_N:\lambda_{\ell+m}^{(N)}=\lambda_\ell^{(N)}\}$. We show that 
\begin{equation}
 \#A_m \le 2N^{d-1} \,.
\end{equation}

We have $A_m = \{\ell \in [\![0,N-1 ]\!]^d : \sum_{j=1}^d \cos(\frac{2\pi(\ell_j+m_j)}{N})-\cos(\frac{2\pi \ell_j}{N}) = 0\}$. Consider the projection of this surface onto a plane. More precisely, suppose $m_j\neq 0$ and consider $P_{e_j^\bot} \ell = (\ell_1,\dots,\ell_{j-1},0,\ell_{j+1},\dots,\ell_d)$. Suppose $n,k \in A_m$ and $P_{e_j^\bot} n = P_{e_j^\bot} k$. Then $n_i = k_i$ for all $i\neq j$. So $\cos(\frac{2\pi(n_j+m_j)}{N})-\cos(\frac{2\pi n_j}{N}) = - \sum_{i\neq j} \cos(\frac{2\pi(n_i+m_i)}{N})-\cos(\frac{2\pi n_i}{N}) = - \sum_{i\neq j} \cos(\frac{2\pi(k_i+m_i)}{N})-\cos(\frac{2\pi k_i}{N}) = \cos(\frac{2\pi(k_j+m_j)}{N})-\cos(\frac{2\pi k_j}{N})$. Since $\cos\theta-\cos\varphi = -2\sin(\frac{\theta+\varphi}{2})\sin(\frac{\theta-\varphi}{2})$, this implies that $\sin\pi(\frac{2n_j+m_j}{N})\sin\frac{\pi m_j}{N} = \sin\pi(\frac{2k_j+m_j}{N})\sin\frac{\pi m_j}{N}$. Since $m_j\in [\![1,N-1 ]\!]$, this implies $\sin\pi(\frac{2n_j+m_j}{N}) = \sin\pi(\frac{2k_j+m_j}{N})$. But $\frac{2n_j+m_j}{N}\le 3$ and $\frac{2k_j+m_j}{N}\le 3$. So we must have $\pi\frac{2n_j+m_j}{N} = \pi\frac{2k_j+m_j}{N}$ or $\pi - \pi\frac{2k_j+m_j}{N}$ or $2\pi + \pi\frac{2k_j+m_j}{N}$. This leads to $n_j = k_j$ or $n_j = \frac{N}{2}-k_j-m_j$ or $n_j = N+k_j$. The last case is excluded as $n_j<N$.

We thus showed that any $(n_1,\dots,n_{j-1},0,n_{j+1},\dots,n_d)$ has at most two preimages within $A_m$ under the mapping $P_{e_j^\bot}$. This implies that $\# A_m \le 2N^{d-1}$ for any $m\neq 0$.

In the special case $a(w) = \ee^{2\pi\ii k\cdot w/N} = N^{d/2} e_k^{(N)}(w)$, we have $a_m^{(N)} = 0$ for $m\neq k$ and $a_k^{(N)} = N^{d/2}$. So the RHS in \eqref{e:limdis} reduces to
\[
\ee^{2\pi\ii k\cdot v/N}\cdot  \frac{\#A_k}{N^d}\to 0 \,.
\]

More generally, suppose $a_N(n) = f(n/N)$ for some $f\in H_s(\T_\ast^d)$, with $s>d/2$. Here $H_s(\T_\ast^d)$ is the Sobolev space of order $s$, with norm $\|f\|_{H^s}^2 = \sum_{k\in \Z^d} |\hat{f}_k|^2\langle k\rangle^{2s}$, where $\hat{f}_k = \int_{\T_\ast^d} \ee^{-2\pi\ii k\cdot x}f(x)\,\dd x$ and $\langle k\rangle = \sqrt{1+|k|^2}$. Then $\|\hat{f}\|_1:=\sum_k |\hat{f}_k| \le C_s \|f\|_{H^s}$, where $C_s^2=\sum_k\langle k\rangle^{-2s}<\infty$ since $2s>d$. On the other hand, $f = \sum_k \hat{f}_k e_k$ with $e_k(x) = \ee^{2\pi\ii k\cdot x}$, so $a_m^{(N)} = \langle e_m^{(N)}, f(\cdot/N)\rangle_{\ell^2(\Lambda_N)} = \sum_{k\in \Z^d} \hat{f}_k \langle e_m^{(N)},e_k(\cdot/N)\rangle_{\ell^2(\Lambda_N)} = \hat{f}_m N^{d/2}$, since $e_k(n/N) = N^{d/2} e_k^{(N)}(n)$.

We showed that $a_m^{(N)} e_m^{(N)}(v) = \hat{f}_m \ee^{2\pi\ii m\cdot v/N}$. Thus,
\[
\left|\lim_{T\to\infty} \frac{1}{T} \int_0^T \langle \ee^{-\ii tA_N}\delta_v,a\ee^{-\ii tA_N}\delta_v\rangle - \langle a\rangle\right|\le \frac{2}{N} \sum_{m\in\Lambda_N} |a_m^{(N)} e_m^{(N)}(v)| \le \frac{2}{N} \|\hat{f}\|_1 \le \frac{C}{N}\|f\|_{H^s}\to 0 .
\]

The estimate is also true if $a_N$ is the restriction to $\Lambda_N$ of some $a\in\ell^1(\Z^d)$. In that case, we have $a = \sum_{n\in \Z^d} c_n \delta_n$ with $\|a\|_1=\sum_{n\in \Z^d} |c_n|<\infty$. On the other hand, $a_m^{(N)} = \sum_{n\in \Z^d} c_n \langle e_m^{(N)},\delta_n\rangle = \sum_{n\in \Z^d} c_n \overline{e_m^{(N)}(n)}$. So $|a_m^{(N)} e_m^{(N)}(v)| \le \frac{1}{N^d} \|a\|_1$, hence $\sum_{m\in \Lambda_N} |a_m^{(N)} e_m^{(N)}(v)| \le \|a\|_1$ and we may conclude as before.

Note that the two cases ($a_N=f(\cdot/N)$ and $a\in \ell^1(\Z^d)$) are distinct, in the sense that $\lim_{N\to\infty} \sum_{n\in \Lambda_N} |f(n/N)| = \infty$ in general.
\end{proof}

\begin{proof}[Proof of Theorem~\ref{thm:zdexten}]
Arguing as before, we find that
\[
\langle \ee^{-\ii tA_N}\delta_v, a\ee^{-\ii tA_N}\delta_w\rangle = \frac{1}{N^d}\sum_{\ell,m\in\Lambda_N} \ee^{\frac{2\pi\ii\ell\cdot(v-w)}{N}}\ee^{\ii t(\lambda_{\ell+m}^{(N)}-\lambda_\ell^{(N)})}a_m^{(N)} e_m^{(N)}(v)
\]
Here, the term $m=0$ is $\frac{1}{N^d}\sum_{\ell\in\Lambda_N} \ee^{\frac{2\pi\ii\ell\cdot(v-w)}{N}} a_0 e_0(v) = 0$ since $v\neq w$. The remaining terms $\sum_{m\neq 0,\ell\in \Lambda_N}$ tend to zero as $T$ followed by $N$ tend to infinity by the same argument as before (the phase $\ee^{\frac{2\pi\ii\ell\cdot(v-w)}{N}}$ makes no difference). This proves the first part.

For the second part, assume $\phi,\psi$ are supported in a compact $K\subset \Lambda_N$, $N$ large enough. Then $\phi = \sum_{v\in K} \phi(v)\delta_v$ and $\psi = \sum_{v\in K} \psi(v)\delta_v$. Thus,
\[
\langle \ee^{-\ii tA_N}\phi, a\ee^{-\ii tA_N}\psi\rangle = \sum_{v,w\in K}  \overline{\phi(w)}\psi(v)\langle \ee^{-\ii tA_N}\delta_v, a\ee^{-\ii tA_N}\delta_w\rangle\,.
\]

Hence,
\begin{multline*}
\langle \ee^{-\ii tA_N}\phi, a\ee^{-\ii tA_N}\psi\rangle - \langle a\rangle\langle \phi,\psi\rangle  = \sum_{v\in K}  \overline{\phi(v)}\psi(v)\left(\langle \ee^{-\ii tA_N}\delta_v, a\ee^{-\ii tA_N}\delta_v\rangle - \langle a\rangle\right) \\
+ \sum_{v,w\in K,v\neq w} \overline{\phi(w)}\psi(v)\langle \ee^{-\ii tA_N}\delta_w, a\ee^{-\ii tA_N}\delta_v\rangle  \,.
\end{multline*}

We see that \eqref{e:ephipsi} follows  from \eqref{e:maindiszd} and \eqref{e:delvw}.
\end{proof}

\begin{proof}[Proof of Proposition~\ref{prp:invertlimzd}]
Take $a_N(n) = f(n/N)$ for $f(x)=\prod_{i=1}^d(1-x_i)$ on $\T_\ast^d$. Then $f\ge 0$. Now
\[
\langle \ee^{-\ii tA_N} \delta_v, a_N \ee^{-\ii tA_N}\delta_v\rangle = \sum_{n\in \Lambda_N} a_N(n) |\ee^{-\ii tA_N}\delta_v(n)|^2 = \sum_{n\in \Z^d} f\Big(\frac{n}{N}\Big) \chi_{\Lambda_N}(n) |\ee^{-\ii tA_N}\delta_v(n)|^2\,,
\]
where we extend $f$ to $\R^d$ arbitrarily in a continuous fashion and we define $\ee^{-\ii tA_N}\delta_v(n):=0$ for $n\notin \Lambda_N$. 

Now $\ee^{-\ii tA_N}\delta_v(w) \to \ee^{-\ii t\cA_{\Z^d}}\delta_v(w)$ for any $w$. This can be seen for example from the explicit expression of the kernels through the Fourier transform, which shows that if $\phi(x)=\sum_{i=1}^d 2\cos 2\pi x_i$ for $x\in \T_\ast^d$, then $\ee^{-\ii tA_N}(w,v)= \frac{1}{N^d}\sum_{n\in \Lambda_N} \ee^{2\pi\ii (w-v)\cdot n/N}\ee^{-\ii t \phi(n/N)}$ for $v,w\in\Lambda_N$ and $\ee^{-\ii t\cA_{\Z^d}}(w,v) = \int_{\T_\ast^d} \ee^{2\pi\ii (w-v)\cdot x}\ee^{-\ii t\phi(x)}\,\dd x$.

We thus have $\lim_{N\to\infty} f(n/N) \chi_{\Lambda_N}(n) |\ee^{-\ii tA_N} \delta_v(n)|^2 = f(0) |\ee^{-\ii t\cA_{\Z^d}}\delta_v(n)|^2$ for any $n$. So by Fatou's lemma,
\begin{equation}\label{e:fatou}
\liminf_{N\to\infty} \langle \ee^{-\ii tA_N}\delta_v, a_N\ee^{-\ii tA_N}\delta_v\rangle \ge \sum_{n\in \Z^d} f(0)|\ee^{-\ii t\cA_{\Z^d}}\delta_v(n)|^2 = 1\,,
\end{equation}
where we used $f(0)=1$ and $\sum_n |\ee^{-\ii t\cA_{\Z^d}}\delta_v(n)|^2 = \|\ee^{\ii t\cA_{\Z^d}} \delta_v\|^2 = \|\delta_v\|^2=1$.

On the other hand, by Riemann integration, $ \langle a_N\rangle = \frac{1}{N^d}\sum_{n\in \Lambda_N} f(n/N) \to \int_{\T_\ast^d} f(x)\,\dd x = \frac{1}{2^d}$. This proves the result.

The statement holds for the average dynamics since $\ee^{-\ii tA_N}\delta_v(w) \to \ee^{-\ii t\cA_{\Z^d}}\delta_v(w)$ implies $\frac{1}{T}\int_0^T |\ee^{-\ii tA_N}\delta_v(w)|^2\,\dd t \to \frac{1}{T}\int_0^T |\ee^{-\ii t\cA_{\Z^d}}\delta_v(w)|^2\,\dd t$, as $|\ee^{-\ii tA_N}\delta_v(w)|^2 \le \|\ee^{-\ii tA_N}\|^2=1$. The (averaged) lower bound \eqref{e:fatou} still holds by Tonelli's theorem.
\end{proof}

\begin{proof}[Proof of Proposition~\ref{prp:noavzd}]
Take $d=1$ and $a_N(x) = \ee^{2\pi\ii x/N}$, so that $a_m = \sqrt{N}\delta_{m,1}$. Since $\lambda_\ell = 2\cos \frac{2\pi \ell}{N}$, \eqref{e:scalex} reduces to $\frac{\ee^{\frac{2\pi\ii v}{N}}}{N}\sum_{\ell=0}^{N-1} \ee^{-4\ii t \sin\frac{\pi}{N}(2\ell+1)\sin\frac{\pi}{N}}$. Specializing to $t=n\in \N$, it is shown\footnote{In \cite[Lemma C.2]{KleiBal} it is assumed the sum takes the form $\Gamma(n) = \sum_{j=1}^t c_j \ee^{2\pi\ii n \theta_j}$ for some distinct $\theta_j\in [0,1)$. In our case, for $N>4$, $|4\sin \frac{\pi}{N}|<\pi$, so $\vartheta_j := \frac{-4\sin\frac{\pi}{N}(2j+1)\sin\frac{\pi}{N}}{2\pi}\in(-\frac{1}{2},\frac{1}{2})$. If $N_1>1$ is the number of distinct $\vartheta_j$, then we may rearrange our sum as $\frac{\ee^{\frac{2\pi\ii v}{N}}}{N}\sum_{\ell=1}^{N_1} c_\ell \ee^{2\pi\ii n\vartheta_\ell}$ and the proof is the same.} in \cite[Lemma C.2]{KleiBal} that this sum has no limit as $n\to\infty$. In particular, $\langle \ee^{-\ii tA_N} \delta_v, a_N\ee^{-\ii tA_N}\delta_v\rangle$ has no limit as $t\to\infty$.

The same lemma shows that $\int_{T-1}^T \langle \ee^{-\ii tA_N} \delta_v, a_N\ee^{-\ii tA_N}\delta_v\rangle\,\dd t$ has no limit as $T\to\infty$. Here the  expression becomes $\frac{\ee^{\frac{2\pi\ii v}{N}}}{N}\sum_{\ell=0}^{N-1}  \ee^{\ii T b_\ell}\frac{(1-\ee^{-\ii b_\ell})}{\ii b_\ell}$ for $b_\ell = -4\sin\frac{\pi}{N}(2\ell+1)\sin\frac{\pi}{N}$.
\end{proof}

\begin{proof}[Proof of Proposition~\ref{prp:ves}]
If $N$ is odd, consider $a_N(n) = 2\cos (\frac{2\pi n_1}{N})$. Then $a_m^{(N)} = N^{d/2}$ if $m = \pm \mathfrak{e}_1$, where $\mathfrak{e}_1=(1,0,\dots,0)$, and $a_m^{(N)}=0$ otherwise. So the RHS \eqref{e:limdis} reduces to
\[
 \ee^{\frac{2\pi\ii v_1}{N}}\frac{\#\{\ell\in\Lambda_N: \lambda_{\ell+\mathfrak{e}_1}^{(N)}=\lambda_\ell^{(N)}\}}{N^d} + \ee^{\frac{-2\pi\ii v_1}{N}} \frac{\#\{\ell\in\Lambda_N: \lambda_{\ell-\mathfrak{e}_1}^{(N)}=\lambda_\ell^{(N)}\}}{N^d}\,.
\]
Since $\lambda_k^{(N)} = \sum_{i=1}^d 2\cos \frac{2\pi k_i}{N}$, we have $\lambda_{\ell+\mathfrak{e}_1}^{(N)} = \lambda_\ell^{(N)}$ iff $\cos(\frac{2\pi(\ell_1+1)}{N}) = \cos(\frac{2\pi\ell_1}{N})$, i.e. $\sin\pi(\frac{2\ell_1+1}{N})\sin\frac{\pi}{N}=0$. This occurs iff $\frac{2\ell_1+1}{N} = 0,1,2$, i.e. $\ell_1 = \frac{-1}{2}$, $\frac{N-1}{2}$ or $\frac{2N-1}{2}$, respectively. The only choice in $\{0,\dots,1\}$ is $\ell_1 = \frac{N-1}{2}$. Since $\ell_j$ can be arbitrary for $j\ge 2$, we see that $\frac{\#\{\ell\in\Lambda_N: \lambda_{\ell+\mathfrak{e}_1}^{(N)}=\lambda_\ell^{(N)}\}}{N^d} = \frac{1}{N}$. Similarly, $\lambda_{\ell-\mathfrak{e}_1}^{(N)} = \lambda_\ell^{(N)}$ iff $\frac{2\ell_1-1}{N} = 0,1,2$, and the only valid choice is $\ell_1 = \frac{N+1}{2}$. We thus showed that the RHS of \eqref{e:limdis} is $b_N(v) = \frac{2\cos(\frac{2\pi v_1}{N})}{N}$.

If $N$ is even, we take $a_N(n) = 2\cos(\frac{4\pi n_1}{N})$. Then $a_m^{(N)}=N^{d/2}$ if $m=\pm 2\mathfrak{e}_1$ and zero otherwise. Here, $\lambda_{\ell\pm 2\mathfrak{e}_1}^{(N)}=\lambda_\ell^{(N)}$ iff $\frac{2\ell_1\pm 2}{N}=0,1,2$, and we conclude as before that $b_N(v) = \frac{4\cos(\frac{4\pi v_1}{N})}{N}$.
\end{proof}

\begin{rem}\label{rem:t=n}
We can take $T$ to depend on $N$, provided it grows fast enough. To see this, back to \eqref{e:scalex}, we notice that in the expansion of $\frac{1}{T}\int_0^T\langle \ee^{\ii tA_N}\delta_v,a\ee^{-\ii tA_N}\delta_v\rangle\,\dd t$, we should now account for the additional term
\begin{equation}\label{e:t=tn}
\frac{1}{N^d}\sum_{\substack{m\in \Lambda_N,\\m\neq 0}} a_m^{(N)} e_m^{(N)}(v)\sum_{\substack{\ell\in\Lambda_N\\\lambda_{\ell+m}^{(N)}\neq \lambda_\ell^{(N)}}}\frac{1}{T}\cdot\frac{\ee^{\ii T(\lambda_{\ell+m}^{(N)}-\lambda_\ell^{(N)})}-1}{\ii(\lambda_{\ell+m}^{(N)}-\lambda_\ell^{(N)})} \,.
\end{equation}

This can be bounded crudely by $\sum_{m\neq 0} |a_m^{(N)}e_m^{(N)}(v)|\cdot \frac{2}{T}\cdot \sup_{\lambda_j^{(N)}\neq \lambda_k^{(N)}} |\lambda_j^{(N)}-\lambda_k^{(N)}|^{-1}$. We showed in the proof that $\sum_{m} |a_m^{(N)} e_m^{(N)}(v)|$ stays bounded for all $N$ for our choice of observables $a$. Thus, if $T = T(N)$ grows faster than the smallest spectral gap between distinct eigenvalues of $A_N$, the term \eqref{e:t=tn} will vanish as required as $N\to\infty$. For $d=1$, it suffices that $T$ grows faster than $N^2$.
\end{rem}

\section{Periodic graphs}\label{sec:pergra}

We here extend ergodicity to $\Z^d$-periodic graphs $\Gamma$. We assume there exist linearly independent vectors $\fa_1,\dots,\fa_d$ in a Euclidean space $\R^D$ such that, if $n_\fa = \sum_{i=1}^d n_i \fa_i$ and $\Z_\fa^d = \{n_\fa:n\in \Z^d\}$, then
\begin{equation}\label{e:vga}
V(\Gamma) = V_f + \Z_\fa^d\,,
\end{equation}
where $V_f$ is the fundamental cell containing a finite number $\nu$ of vertices, which is then repeated periodically under translations by $n_\fa\in \Z_\fa^d$.

For example, $\Gamma=\Z^d$ has $V_f=\{0\}$ and $\fa_j=\mathfrak{e}_j$ the standard basis. An infinite strip of width $k$ has $V_f = P_k$, the $k$-path, $d=1$ and $\fa_1=\mathfrak{e}_1$. See \cite{McKSa,SaYo} for more examples.

We endow $V_f$ with a potential $(Q_1,\dots,Q_\nu)$ and copy these values across the blocks $V_f+n_\fa$. This turns $Q$ into a periodic potential on $\Gamma$. We consider the Schr\"odinger operator $\cH = \cA_\Gamma+Q$.

From \eqref{e:vga}, any $u\in\Gamma$ takes the form $u=u_\fa+\{u\}_\fa$ for some $u_\fa\in\Z_\fa^d$ and $\{u\}_\fa\in V_f$.

Fix a large $N$ and let $\Gamma_N = \cup_{n\in \LL_N^d} (V_f+n_\fa)$, where $\LL_N^d = \{0,\dots,N-1\}^d$. We consider the restriction $H_N$ on $\Gamma_N$ with periodic boundary conditions. Then it holds that \cite{McKSa}, if $U:\ell^2(\Gamma_N)\to\mathop\oplus_{j\in\mathbb{L}_N^d} \ell^2(V_f)$ is defined by
\[
(U\psi)_j(v_i) = \frac{1}{N^{d/2}}\sum_{k\in \Lambda_N}\ee^{\frac{-2\pi\ii j\cdot k}{N}}\psi(v_i+k_\fa)\,,
\]
then $U$ is unitary and
\begin{equation}\label{e:dia}
UH_NU^{-1} = \mathop\oplus_{j\in\LL_N^d} H\Big(\frac{j_\fb}{N}\Big)\,,
\end{equation}
where $b_1,\dots,b_\nu$ is the dual basis of $(\fa_i)$ satisfying $\fa_i\cdot \fb_j = 2\pi \delta_{i,j}$, $n_\fb = \sum_{i=1}^d n_i \fb_i$ and
\[
H(\theta_\fb)f(v_i) = \sum_{u\sim v_i}\ee^{\ii\theta_\fb\cdot \lfloor u\rfloor_\fa}f(\{u\}_\fa) + Q_if(v_i) \,.
\]
Much like $\cA_{\Z^d}$ is unitarily equivalent to multiplication by a function that is a sum of cosines via the Fourier transform, \eqref{e:dia} is a finite version of the fact that $\cH$ is unitarily equivalent to multiplication by a $\nu\times \nu$ matrix $H(\theta_\fb)$ via the Floquet transform $U$. Denote by $E_s(\theta_\fb)$, $s=1,\dots,\nu$ the eigenvalues of $H(\theta_\fb)$.

As initial state, we consider $\psi = \delta_{v_p}\otimes\delta_{n_\fa}$, more precisely $\psi(v_i+k_\fa):= \delta_{v_p}(v_i)\delta_{n}(k)$ for $v_i\in V_f$ and $k\in \Z^d$. In other words, we start from a point mass. 

\begin{thm}\label{thm:pergra}
Assume that
\begin{equation}\label{e:flo}
\sup_{m\neq 0} \frac{\#\{(r,s,w)\in \LL_N^d\times \{1,\dots,\nu\}^2:E_s(\frac{r_{\fb}+m_{\fb}}{N})-E_w(\frac{r_{\fb}}{N})=0\}}{N^d}\to 0
\end{equation}
as $N\to\infty$. Suppose the observable $a_N$ satisfies one of the following conditions:
\begin{enumerate}[\rm(i)]
\item $a_N(k_\fa+v_q) = f^{(q)}(k/N)$ for some $\nu$ functions $f^{(q)}\in H^s(\T_\ast^d)$, with $s>d/2$,
\item or, $a_N$ is the restriction to $\Gamma_N$ of an integrable function $a\in\ell^1(\Gamma)$.
\end{enumerate}
Then
\[
\lim_{N\to\infty}\Big|\lim_{T\to\infty}\frac{1}{T}\int_0^T\langle \ee^{-\ii tH_N}\delta_{v_p}\otimes\delta_{n_\fa}, a\ee^{-\ii tH_N}\delta_{v_p}\otimes \delta_{n_\fa}\rangle \\
- \langle a\rangle_p\Big|=0\,,
\]
where, denoting $\langle a(\cdot+v_q)\rangle:= \frac{1}{N^d} \sum_{n\in \LL_N^d} a(n_\fa+v_q)$, 
\begin{equation}\label{e:avp}
\langle a\rangle_p = \frac{1}{N^d}\sum_{r\in\LL_N^d} \sum_{q=1}^\nu \langle a(\cdot+v_q)\rangle \sum_{s=1}^{\nu'} \Big|\Big[P_{E_s}\Big(\frac{r_{\fb}}{N}\Big)\delta_{v_q}\Big](v_p)\Big|^2\,.
\end{equation}
\end{thm}
The Floquet condition \eqref{e:flo} was used as a requirement for quantum ergodicity in \cite{McKSa}. It is a bit stronger than asking that $H$ has purely absolutely continuous spectrum. We refer to \cite{McKSa} for numerous examples which satisfy \eqref{e:flo}.

In the special case $\langle a(\cdot+v_q)\rangle = \langle a(\cdot+v_1)\rangle$ $\forall q=1,\dots,\nu$, \eqref{e:avp} reduces to $\langle a(\cdot+v_1)\rangle$. In fact, we get
\[
\frac{1}{N^d}\langle a(\cdot+v_1)\rangle\sum_{r\in\LL_N^d}\sum_{s=1}^{\nu'} \Big\|P_{E_s}\Big(\frac{r_\fb}{N}\Big)\delta_{v_p}\Big\|^2 = \frac{1}{N^d} \langle a(\cdot+v_1)\rangle\sum_{r\in\LL_N^d} \|\delta_{v_p}\|^2= \langle a(\cdot+v_1)\rangle\,.
\]
This scenario occurs in particular if $a$ is locally constant, i.e. takes a fixed value on each periodic block $V_f+n_\fa$, which depends on $n$ but not on $v_q\in V_f$.

In general, \eqref{e:avp} gives not the uniform average of $a$, but a weighted average, with weights depending on $p$ and the spectral decomposition of the Floquet matrix. Note however that $\frac{1}{\nu}\sum_{p=1}^\nu \langle a\rangle_p = \frac{1}{\nu}\sum_{q=1}^\nu \langle a(\cdot+v_q)\rangle$ is the uniform average. So the mean density $\mu_{n,T}^{(N)}(u)=\frac{1}{T}\int_0^T \frac{1}{\nu}\sum_{p=1}^\nu |(\ee^{-\ii tH_N} \delta_{v_p}\otimes \delta_{n_\fa})(u)|^2\,\dd t$ on $\Gamma_N$ approaches the uniform measure $\frac{1}{\nu N^d}$ for $T,N\gg 0$.

\begin{exa}\label{exa:per}
Let $Q\equiv 0$. The average $\langle a\rangle_p$ is the uniform average if:
\begin{enumerate}[(i)]
\item $\nu=1$, for example $\Gamma=\Z^d$ or the triangular lattice. See \cite[\S~4.1]{McKSa} for more examples.
\item $\Gamma$ is the hexagonal lattice or $\Gamma$ is an infinite ladder (strip of width $2$). The argument is given in \cite[\S~4.2,4.3]{McKSa}.
\end{enumerate}
If $\Gamma$ is the infinite strip of width $3$, then $\langle a\rangle_p$ is not the uniform average. Here, $A(\theta_\fb) = \begin{pmatrix} c_\theta& 1&0\\ 1&c_\theta&1\\ 0&1&c_\theta\end{pmatrix}$ for $c_\theta = 2\cos 2\pi\theta$. The eigenvectors are independent of $\theta$ and given by $w_1 = \frac{1}{2}(1,\sqrt{2},1)$, $w_2=\frac{1}{\sqrt{2}}(-1,0,1)$ and $w_3 = \frac{1}{2}(1,-\sqrt{2},1)$ for $\lambda_1=c_\theta+\sqrt{2}$, $\lambda_2 = c_\theta$, $\lambda_3 = c_\theta-\sqrt{2}$. It follows that $(P_i\delta_{v_q})(v_p) = w_i(v_p)w_i(v_q)$.

See Figure~\ref{fig:strips}. Suppose we take $v_p=v_1$. Then
\[
\sum_{i=1}^3 |(P_i \delta_{v_q})(v_1)|^2 = \frac{|w_1(v_q)|^2}{4} + \frac{|w_2(v_q)|^2}{2} + \frac{|w_3(v_q)|^2}{4}\,.
\]
For $q=1,2,3$, this gives $\frac{3}{8}$, $\frac{1}{4}$ and $\frac{3}{8}$, respectively. So $\langle a\rangle_1 = \frac{3\langle a(\cdot+v_1)\rangle + 2\langle a(\cdot +v_2)\rangle + 3\langle a(\cdot + v_3)\rangle}{8}$, which is not the uniform average: there is more weight to both sides of the strip.

For comparison, suppose we take $v_p=v_2$, the central vertex. Then
\[
\sum_{i=1}^3 |(P_i \delta_{v_q})(v_2)|^2 = \frac{|w_1(v_q)|^2}{2} + \frac{|w_3(v_q)|^2}{2}\,.
\]
For $q=1,2,3$, this gives $\frac{1}{4}$, $\frac{1}{2}$ and $\frac{1}{4}$, respectively. So $\langle a\rangle_2 = \frac{\langle a(\cdot+v_1)\rangle + 2\langle a(\cdot +v_2)\rangle + \langle a(\cdot + v_3)\rangle}{4}$, which is not the uniform average either. There is more weight to the center of the strip.

More surprisingly perhaps, the spreading is not uniform in cylinders either, which are regular, very homogeneous graphs. 

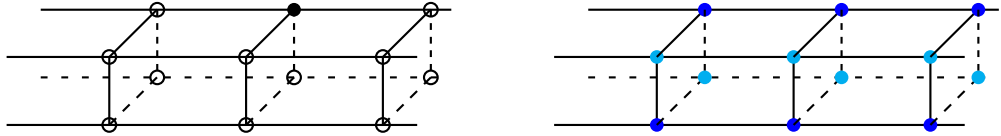
\begin{figure}[h!]
\begin{center}
\setlength{\unitlength}{0.9cm}
\thicklines
\begin{picture}(1.3,1.4)(-0.8,-0.8)
	 \put(-7,0){\line(1,0){6}}
	 \put(-7,-1){\line(1,0){6}}
	 \put(-5.5,-1){\line(0,1){1}}
	 \put(-3.5,-1){\line(0,1){1}}
	 \put(-1.5,-1){\line(0,1){1}}
	 \put(-3.5,-1){\circle{.2}}
	 \put(-3.5,0){\circle{.2}}
	 \put(-5.5,-1){\circle{.2}}
	 \put(-5.5,0){\circle{.2}}
	 \put(-1.5,-1){\circle{.2}}
	 \put(-1.5,0){\circle{.2}}
	 \put(-5.5,0){\line(1,1){0.7}}
	 \put(-3.5,0){\line(1,1){0.7}}
	 \put(-1.5,0){\line(1,1){0.7}}
	 \put(-6.5,0.7){\line(1,0){6}}
	 \multiput(-6.5,-0.3)(0.3,0){20}{\line(1,0){0.1}}
	 \put(-4.8,0.7){\circle{.2}}
	 \put(-2.8,0.7){\circle*{.2}}
	 \put(-0.8,0.7){\circle{.2}}
 	 \multiput(-4.8,0.7)(0,-0.2){5}{\line(0,-1){0.1}}
  \multiput(-2.8,0.7)(0,-0.2){5}{\line(0,-1){0.1}}
	 \multiput(-0.8,0.7)(0,-0.2){5}{\line(0,-1){0.1}}
	 \multiput(-5.5,-1)(0.2,0.2){4}{\line(1,1){0.1}}
	\multiput(-3.5,-1)(0.2,0.2){4}{\line(1,1){0.1}}
	\multiput(-1.5,-1)(0.2,0.2){4}{\line(1,1){0.1}}
	\put(-4.8,-0.3){\circle{.2}}
	\put(-2.8,-0.3){\circle{.2}}
	\put(-0.8,-0.3){\circle{.2}}
	 \put(1,0){\line(1,0){6}}
	 \put(1,-1){\line(1,0){6}}
	 \put(2.5,-1){\line(0,1){1}}
	 \put(4.5,-1){\line(0,1){1}}
	 \put(6.5,-1){\line(0,1){1}}
	 \put(4.5,-1){\textcolor{blue}{\circle*{.2}}}
	 \put(4.5,0){\textcolor{cyan}{\circle*{.2}}}
	 \put(2.5,-1){\textcolor{blue}{\circle*{.2}}}
	 \put(2.5,0){\textcolor{cyan}{\circle*{.2}}}
	 \put(6.5,-1){\textcolor{blue}{\circle*{.2}}}
	 \put(6.5,0){\textcolor{cyan}{\circle*{.2}}}
	 \put(2.5,0){\line(1,1){0.7}}
	 \put(4.5,0){\line(1,1){0.7}}
	 \put(6.5,0){\line(1,1){0.7}}
	 \put(1.5,0.7){\line(1,0){6}}
	 \multiput(1.5,-0.3)(0.3,0){20}{\line(1,0){0.1}}
	 \put(3.2,0.7){\textcolor{blue}{\circle*{.2}}}
	 \put(5.2,0.7){\textcolor{blue}{\circle*{.2}}}
	 \put(7.2,0.7){\textcolor{blue}{\circle*{.2}}}
 	 \multiput(3.2,0.7)(0,-0.2){5}{\line(0,-1){0.1}}
  \multiput(5.2,0.7)(0,-0.2){5}{\line(0,-1){0.1}}
	 \multiput(7.2,0.7)(0,-0.2){5}{\line(0,-1){0.1}}
	 \multiput(2.5,-1)(0.2,0.2){4}{\line(1,1){0.1}}
	\multiput(4.5,-1)(0.2,0.2){4}{\line(1,1){0.1}}
	\multiput(6.5,-1)(0.2,0.2){4}{\line(1,1){0.1}}
	\put(3.2,-0.3){\textcolor{cyan}{\circle*{.2}}}
	\put(5.2,-0.3){\textcolor{cyan}{\circle*{.2}}}
	\put(7.2,-0.3){\textcolor{cyan}{\circle*{.2}}}
\end{picture}
\caption{A point mass (left) spreads $\frac{3}{8}$ of its mass over both its line and the line diagonally opposite to it, and only $\frac{1}{8}$ of its mass on each of the other two lines (right). If the cylinder has size $4N$, then each dark blue vertex carries a mass $\frac{3}{8N}$ and each light blue vertex carries a mass $\frac{1}{8N}$.}\label{fig:carte}
\end{center}
\end{figure}

For example, for the $4$-cylinder in Figure~\ref{fig:carte}, we have $A(\theta_\fb) = c_\theta \mathrm{Id}_4 + \cA_{C_4}$, where $C_4$ is the $4$-cycle, so $A(\theta_\fb)$ shares the eigenvectors of $\cA_{C_4}$ given by
\[
\textstyle 2,\ 0,\ 0,\ -2, \qquad \frac{1}{2}(1, 1, 1, 1), \ \frac{1}{\sqrt{2}}(0, -1, 0, 1), \ \frac{1}{\sqrt{2}}(-1, 0, 1, 0), \ \frac{1}{2}(-1, 1, -1, 1),
\] 
respectively. If $w_i$ are the eigenvectors in this order, then the three eigenprojections are again independent of $\theta$ (this holds in general for Cartesian products such as $\Z^d\mathop\square G_F$, with $G_F$ finite) and given by $(P_{E_1}\delta_{v_q})(v_p) = w_1(v_p)w_1(v_q)$, $(P_{E_3} \delta_{v_q})(v_p) = w_4(v_p)w_4(v_q)$ and $(P_{E_2}\delta_{v_q})(v_p) = w_2(v_p)w_2(v_q)+w_3(v_p)w_3(v_q)$. Hence, $|(P_{E_1}\delta_{v_q})(v_p)|^2 = |(P_{E_3}\delta_{v_q})(v_p)|^2 = \frac{1}{16}$. We may assume $v_p=v_1$ by homogeneity. Then $\sum_{i=1}^3 |(P_{E_i}\delta_{v_q}(v_1)|^2 = \frac{1}{8} + \frac{|w_3(v_q)|^2}{2}$. For $q=1,2,3,4$, this gives $\frac{3}{8}$, $\frac{1}{8}$, $\frac{3}{8}$ and $\frac{1}{8}$, respectively. This is illustrated in Figure~\ref{fig:carte}.

It was observed in \cite{McKSa} that some eigenbases of the cylinder are uniformly distributed while others are not. We see that having one equidistributed eigenbasis is not enough to obtain the dynamic equidistribution that we discuss in this paper. This is in contrast to the folklore physics heuristics of \S~\ref{sec:et}.
\end{exa}

\begin{proof}[Proof of Theorem~\ref{e:avp}]
It is shown in \cite[Lemma 2.2]{McKSa} that
\[
\frac{1}{T}\int_0^T \ee^{\ii tH_N}a\ee^{-\ii tH_N}\,\dd t \psi(k_{\fa}+v_i) = \sum_{r\in\LL_N^d}	\sum_{\ell=1}^\nu (U\psi)_r(v_\ell)  F_T(k,r;v_i,v_\ell) e_r^{(N)}(k)\,,
\]
where
\begin{multline}\label{e:ft}
F_T(k,r;v_i,v_\ell): = \sum_{m\in \LL_N^d}\sum_{q,s,w=1}^\nu  \frac{1}{T}\int_0^T\ee^{\ii t[E_s(\frac{r_{\fb}+m_{\fb}}{N})-E_w(\frac{r_{\fb}}{N})]}\,\dd t \, \\
\times P_s\Big(\frac{r_{\fb}+m_{\fb}}{N}\Big)(v_i,v_q)a_m^{(N)}(v_q)P_w\Big(\frac{r_{\fb}}{N}\Big)(v_q,v_\ell)e_{m}^{(N)}(k)\,,
\end{multline}
and $a_m^{(N)}(v_q) = \langle e_m^{(N)},a(\cdot_\fa+v_q)\rangle = \sum_{n\in\LL_N^d} \overline{e_m^{(N)}(n)}a(n_\fa+v_q)$.

If $\psi = \delta_{v_p}\otimes\delta_{n_\fa}$, then $(U\psi)_r(v_\ell) = \frac{1}{N^{d/2}} \delta_{v_p}(v_\ell) \ee^{\frac{-2\pi\ii r\cdot n}{N}}$. Hence,
\begin{equation}\label{e:ipe}
\frac{1}{T}\int_0^T \ee^{\ii tH_N}a\ee^{-\ii tH_N}\,\dd t \delta_{v_p}\otimes \delta_{n_\fa}(k_{\fa}+v_i) = \frac{1}{N^{d/2}}\sum_{r\in\LL_N^d} \ee^{\frac{-2\pi\ii r\cdot n}{N}}  F_T(k,r;v_i,v_p) e_r^{(N)}(k)\,.
\end{equation}
Since $\langle A\delta_{v_p}\otimes \delta_{n_\fa},B \delta_{v_p}\otimes \delta_{n_\fa}\rangle = (A^\ast B\delta_{v_p}\otimes \delta_{n_\fa})(v_p+n_\fa)$, we consider
\[
\frac{1}{T}\int_0^T \ee^{\ii tH_N}a\ee^{-\ii tH_N}\,\dd t \delta_{v_p}\otimes \delta_{n_\fa}(n_{\fa}+v_p) = \frac{1}{N^{d}}\sum_{r\in\LL_N^d}  F_T(n,r;v_p,v_p)\,.
\]

Taking the limit $T\to\infty$, this reduces to \cite{McKSa},
\begin{equation}\label{e:meanbkrp}
\frac{1}{N^{d}}\sum_{r\in\LL_N^d}  b(n,r;v_p,v_p)
\end{equation}
where, denoting $S_r = \{(m,s,w):E_s(\frac{r_\fb+m_\fb}{N})-E_w(\frac{r_\fb}{N})=0\}$,  we have
\begin{multline}\label{e:blim}
b(n,r,v_i,v_\ell) = \sum_{m\in \LL_N^d}\sum_{q,s,w=1}^\nu \mathbf{1}_{S_r}(m,s,w)P_s\Big(\frac{r_{\fb}+m_{\fb}}{N}\Big)(v_i,v_q) \\
\times a_m^{(N)}(v_q) P_w\Big(\frac{r_{\fb}}{N}\Big)(v_q,v_\ell)e_{m}^{(N)}(n)\,,
\end{multline}
If in \eqref{e:meanbkrp} we consider only the term $m=0$ from \eqref{e:blim}, with $v_i=v_\ell=v_p$, we get
\begin{multline}\label{e:m=0}
\frac{1}{N^d}\sum_{r\in\LL_N^d}\sum_{\substack{q,s,w=1\\E_s=E_w}}^\nu  P_s\Big(\frac{r_{\fb}}{N}\Big)(v_p,v_q)a_0^{(N)}(v_q)P_w\Big(\frac{r_{\fb}}{N}\Big)(v_q,v_p)e_{0}^{(N)}(n)\\
=\frac{1}{N^d}\sum_{r\in\LL_N^d} \sum_{q=1}^\nu \langle a(\cdot+v_q)\rangle \sum_{s=1}^{\nu'} P_{E_s}\Big(\frac{r_{\fb}}{N}\Big)(v_p,v_q)P_{E_s}\Big(\frac{r_{\fb}}{N}\Big)(v_q,v_p) = \langle a\rangle_p
\end{multline}
where $\nu'$ is the number of distinct eigenvalues. 
To prove the theorem, we should show that
\[
\frac{1}{N^d}\sum_{r\in\LL_N^d} \sum_{m\neq 0}\sum_{q,s,w=1}^\nu \mathbf{1}_{S_r}(m,s,w)P_s\Big(\frac{r_{\fb}+m_{\fb}}{N}\Big)(v_p,v_q) a_m^{(N)}(v_q) P_w\Big(\frac{r_{\fb}}{N}\Big)(v_q,v_p)e_{m}^{(N)}(n)\to 0
\]
Let $A_m = \{(r,s,w):E_s(\frac{r_{\fb}+m_{\fb}}{N})-E_w(\frac{r_{\fb}}{N})=0\}$. Then $(m,s,w)\in S_r\iff (r,s,w)\in A_m$ so the above is
\[
\frac{1}{N^d} \sum_{m\neq 0}\sum_{q=1}^\nu a_m^{(N)}(v_q)e_{m}^{(N)}(n) \sum_{r\in\LL_N^d}\sum_{s,w=1}^{
\nu} \mathbf{1}_{A_m}(r,s,w)P_s\Big(\frac{r_{\fb}+m_{\fb}}{N}\Big)(v_p,v_q) P_w\Big(\frac{r_{\fb}}{N}\Big)(v_q,v_p)\,.
\]
Assume $\sum_m\sum_{q=1}^\nu |a_m^{(N)}(v_q)e_m^{(N)}(n)|\le C_a$ (observable condition). By \eqref{e:flo}, $\sup_{m\neq 0} \frac{|A_m|}{N^d}\to 0$. Hence, the above tends to $0$ as required, since $|P_s(\theta_\fb)(v,w)|\le 1$. 

The observable condition is satisfied for the two classes we have. If $a_N(k_\fa+v_q) = f^{(q)}(k/N)$ with $f^{(q)}\in H^s(\T_\ast^d)$, $s>d/2$, then $a_m^{(N)}(v_q) = \langle e_m^{(N)},f^{(q)}(\cdot/N)\rangle_{\ell^2(\LL_N^d)} = \hat{f}_m^{(q)}N^{d/2}$. As before, this implies that $\sum_{m}\sum_q |a_m^{(N)}(v_q)e_m^{(N)}(n)|\le \sum_{q} \|\hat{f}^{(q)}\|_1$, which is finite.

The second scenario is that $a_N$ is the restriction to $\Gamma_N$ of some $a\in\ell^1(\Gamma)$. Here, $a = \sum_{n\in \Z^d}\sum_{q=1}^\nu c_{n,q}\delta_{n_\fa+v_q}$ with $\sum_{n,q}|c_{n,q}|<\infty$. Then $a_m^{(N)}(v_q) = \sum_{n\in \Z^d} c_{n,q} \overline{e_m^{(N)}(n)}$. This implies $|a_m^{(N)}(v_q)e_m^{(N)}(n)|\le \frac{1}{N^d}\|a\|_1$ for all $q$ implying the hypothesis.
\end{proof}

\begin{rem}[Another natural initial state]\label{rem:perini}
We may ask what happens if instead of starting from a point mass $\delta_{v_p}\otimes\delta_{n_\fa}$, our initial state is equally distributed on the fundamental set, that is $\psi_0 = \frac{1}{\sqrt{\nu}} \mathbf{1}_{V_f}\otimes \delta_{n_\fa}$ for some fixed $n\in \LL_N^d$. In case of the ladder for instance, this corresponds to a vector localized on the two vertices of $V_f$, each carrying mass $\frac{1}{\sqrt{2}}$. We will see that the limiting distribution is still not the uniform average in general. 

Revisiting the proof, we now have $(U\psi_0)_r(v_\ell) = \frac{1}{\sqrt{\nu}N^{d/2}} \ee^{\frac{-2\pi\ii r\cdot n}{N}}$, so the RHS of \eqref{e:ipe} becomes $\frac{1}{\sqrt{\nu}N^{d/2}}\sum_r \ee^{\frac{-2\pi\ii r\cdot n}{N}}\sum_{\ell=1}^\nu F_T(k,r,v_i,v_\ell) e_r^{(N)}(k)$. Here, we have$\langle \frac{1}{\sqrt{\nu}} \mathbf{1}_{V_f}\otimes \delta_{n_\fa},\phi\rangle = \frac{1}{\sqrt{\nu}}\sum_{i=1}^\nu \phi(n+v_i)$, so \eqref{e:meanbkrp} is replaced by $\frac{1}{\nu N^d} \sum_r\sum_{\ell,i=1}^\nu b(n,r,v_i,v_\ell)$. Consequently, instead of \eqref{e:m=0} we get $\frac{1}{\nu N^d}\sum_r \sum_{\ell,i=1}^\nu \sum_{q=1}^\nu \langle a(\cdot+v_q)\rangle \sum_{s=1}^{\nu'} P_{E_s}\big(\frac{r_{\fb}}{N}\big)(v_i,v_q)P_{E_s}\big(\frac{r_{\fb}}{N}\big)(v_q,v_\ell)$. This simplifies to
\[
\mathbf{E}(a)=\frac{1}{\nu N^d} \sum_{r\in\LL_N^d} \sum_{q=1}^\nu  \langle a(\cdot+v_q)\rangle\sum_{s=1}^{\nu'}\Big|\Big[P_{E_s}\Big(\frac{r_\fb}{N}\Big)\mathbf{1}_{V_f}\Big](v_q)\Big|^2\,.
\]
The rest of the proof is the same, so our theorem now says that averaging $a$ over the evolution of $\psi_0$ is close to $\mathbf{E}(a)$. Comparing with Example~\ref{exa:per}, in case of the ladder and the honeycomb lattice, this is again the uniform average. In case of the strip of width $3$, we here have $(P_i\mathbf{1}_{V_f})(v_q) = \langle w_i,\mathbf{1}_{V_f}\rangle w_i(v_q)$, so
\[
\sum_{i=1}^3 |(P_i\mathbf{1}_{V_f})(v_q)|^2 = \frac{(2+\sqrt{2})^2|w_1(v_q)|^2}{4} + \frac{(2-\sqrt{2})^2|w_3(v_q)|^2}{4}\,.
\]
For $q=1,2,3$, this gives $\frac{(2+\sqrt{2})^2+(2-\sqrt{2})^2}{16}=\frac{3}{4}$, $\frac{(2+\sqrt{2})^2+(2-\sqrt{2})^2}{8}=\frac{3}{2}$ and $\frac{3}{4}$, respectively. Thus, $\mathbf{E}(a)=\frac{1}{3}\cdot\frac{3\langle a(\cdot+v_1)\rangle + 6\langle a(\cdot+v_2)\rangle +3\langle a(\cdot+v_3)\rangle}{4} = \frac{\langle a(\cdot+v_1)\rangle + 2\langle a(\cdot+v_2)\rangle +\langle a(\cdot+v_3)\rangle}{4}$. So we still don't get the uniform average; there is more weight given to the middle line. Curiously, this is the same as starting from a point mass in the middle.

In case of the cylinder, $P_{E_1} \mathbf{1}_{V_f} = \langle w_1,\mathbf{1}_{V_f}\rangle w_1 = 2 w_1$, while $P_{E_2}\mathbf{1}_{V_f} = P_{E_3} \mathbf{1}_{V_f}=0$, since $w_2,w_3,w_4$ are all orthogonal to $\mathbf{1}_{V_f}$. It follows that $\sum_{i=1}^4 |P_{E_s} \mathbf{1}_{V_f}(v_q)|^2 = 4 |w_1(v_q)|^2 = 1$. Thus, $\mathbf{E}(a) = \frac{1}{4}\sum_{q=1}^4 \langle a(\cdot+v_q)\rangle$ is now the uniform average, in contrast to the case of an initial state consisting of a point mass which was discussed in Example~\ref{exa:per}.
\end{rem}

In general, if the initial state $\psi_0$ has a compact support, the limiting average becomes
\begin{equation}\label{e:gena}
\mathbf{E}_{\psi_0}(a) = \sum_{r\in\LL_N^d} \sum_{q=1}^\nu \langle a(\cdot+v_q)\rangle \sum_{s=1}^{\nu'} \Big| \Big[P_{E_s}\Big(\frac{r_\fb}{N}\Big)(U\psi_0)_r\Big](v_q)\Big|^2\,.
\end{equation}

Regarding the Floquet assumption \eqref{e:flo}, it is likely to be necessary in view of \cite[Prp. 1.6]{McKSa}. It is clear that it cannot be completely dropped, as this would allow the presence of ``flat bands'', that is infinitely degenerate eigenvalues for $\cH$ with eigenvectors of compact support. If we take such an eigenvector as an initial state, it will not spread, since we simply get $\ee^{-\ii tH_N} \psi_0 = \ee^{-\ii t \lambda} \psi_0$, so $|\ee^{-\ii tH_N} \psi_0| = |\psi_0|$ for all times. An example is given in Figure~\ref{fig:fla}. See \cite{McKSa,SaYo} for more background on this phenomenon. Hence, at least pure AC spectrum for $\cH$ should be assumed, but \eqref{e:flo} is stronger than this.

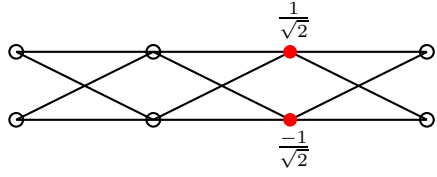
\begin{figure}[h!]
\begin{center}
\setlength{\unitlength}{0.9cm}
\thicklines
\begin{picture}(1.3,1.8)(-0.8,-1.4)
	 \put(-3,0){\line(1,0){6}}
	 \put(-3,-1){\line(1,0){6}}
	 \put(-3,-1){\line(2,1){2}}
	 \put(-3,0){\line(2,-1){2}}
	 \put(-1,-1){\line(2,1){2}}
	 \put(-1,0){\line(2,-1){2}}
	 \put(1,-1){\line(2,1){2}}
	 \put(1,0){\line(2,-1){2}}
	 \put(1,-1){\textcolor{red}{\circle*{.2}}}
	 \put(1,0){\textcolor{red}{\circle*{.2}}}
	 \put(-1,-1){\circle{.2}}
	 \put(-1,0){\circle{.2}}
	 \put(-3,-1){\circle{.2}}
	 \put(-3,0){\circle{.2}}
	 \put(3,-1){\circle{.2}}
	 \put(3,0){\circle{.2}}
	 \put(0.8,0.4){\small{$\frac{1}{\sqrt{2}}$}}
	 \put(0.8,-1.5){\small{$\frac{-1}{\sqrt{2}}$}}
\end{picture}
\caption{An initial state with the given weights (and zero on the remaining vertices) stays frozen and does not spread under the action of $\ee^{-\ii tA_N}$. This graph has a flat band $\lambda=0$.}\label{fig:fla}
\end{center}
\end{figure}

\section{Continuous case}

\subsection{Regularizing in momentum space}\label{sec:conmo}
The Dirac distribution $\delta_y$ on $\R^d$ satisfies $\langle \delta_y,f\rangle = f(y)$. As in \cite{AR12}, we consider here a normalized truncated Dirac distribution defined by $\delta_y^I:=\frac{1}{\sqrt{N_I}} \mathbf{1}_I(-\Delta)\delta_y$, where $I$ is an interval and $N_I$ is the number of eigenvalues of $-\Delta$ in $I$. Let us fix $I=(-\infty,E]$ and denote $N_E= N_I$, $\delta_y^E = \delta_y^I$ and $\mathbf{1}_{\le E} = \mathbf{1}_I$.

In our framework, $\delta_y^E$ is a trigonometric polynomial, as we can see by \emph{defining} $\delta_y^E$ through its Fourier expansion, $\delta_y^E:= \sum_j  \langle e_j, \delta_y^E\rangle e_j = \frac{1}{\sqrt{N_E}}\sum_{\lambda_j\le E} \overline{e_j(y)} e_j$, for $e_j(x)=\ee^{2\pi\ii j\cdot x}$. This function satisfies $\delta_y^E(y) = \sqrt{N_E} \to \infty$ as $E\to \infty$ and $\delta_y^E(x) \to 0$ as $E\to \infty$ for $x\neq y\in \T_\ast^d$ (see the \emph{proof of (2)} below). Also, $\|\delta_y^E\|^2  = \frac{1}{N_E} \sum_{\lambda_j,\lambda_k\le E} e_{k-j}(y)\langle e_j,e_k\rangle = 1$ and $\langle \delta_y^E,f\rangle = \frac{1}{\sqrt{N_E}} \sum_{\lambda_j\le E} e_j(y) \langle e_j,f\rangle = \frac{1}{\sqrt{N_E}} [\mathbf{1}_{\le E}(-\Delta) f](y)$.

\begin{proof}[Proof of Theorem~\ref{thm:ctstor}]
We first note that if $a= \sum_m a_m e_m$, then
\begin{equation}\label{e:obsob}
\sum_{m\in \Z^d} |a_m|<\infty
\end{equation}
since we assumed that $\|a\|_{H_s}^2 = \sum_m  |a_m|^2 \langle m\rangle^{2s} <\infty$ for $\langle m\rangle = \sqrt{1+m^2}$ and $s>d/2$.

We have $\ee^{\ii t \Delta} \delta_y^E = \frac{1}{\sqrt{N_E}}\sum_{\lambda_\ell \le E} \overline{e_\ell(y)} \ee^{-\ii t\lambda_\ell} e_\ell$ and $a = \sum_m a_m e_m$. As $e_me_\ell = e_{m+\ell}$, we get $\ee^{-\ii t \Delta} a\ee^{\ii t\Delta} \delta_y^E = \frac{1}{\sqrt{N_E}} \sum_m a_m \sum_{\lambda_\ell \le E} \overline{e_\ell (y)} \ee^{\ii t(\lambda_{\ell+m}-\lambda_\ell)} e_{m+\ell}$.

Thus, $\langle \delta_y^E, \ee^{-\ii t\Delta} a \ee^{\ii t\Delta} \delta_y^E\rangle = \frac{1}{N_E} \sum\limits_{\substack{m,\ell\in \Z^d,\\ \lambda_\ell,\lambda_{m+\ell}\le E}} a_m \ee^{\ii t(\lambda_{\ell+m}-\lambda_\ell)}e_m(y)$. The term $m=0$ corresponds to $\frac{1}{N_E} \sum_{\lambda_\ell \le E} a_0 e_0(y) = \int_{\T_\ast^d} a(x)\,\dd x$. Since $\frac{1}{T}\int_0^T \ee^{\ii t(\lambda_{\ell+m}-\lambda_\ell)}\,\dd t = \frac{1}{T}\cdot \frac{\ee^{\ii T(\lambda_{\ell+m}-\lambda_\ell)}-1}{\ii(\lambda_{\ell+m}-\lambda_\ell)}$ for $\lambda_{\ell+m}\neq \lambda_\ell$, we get
\begin{multline}\label{e:conclucont}
 \frac{1}{T}\int_0^T \langle \ee^{\ii t\Delta} \delta_y^E, a \ee^{\ii t\Delta} \delta_y^E\rangle\,\dd t = \int_{\T_\ast^d} a(x)\,\dd x + \frac{1}{N_E}\sum_{\substack{m\neq 0,\ell\in \Z^d\\ \lambda_\ell,\lambda_{\ell+m}\le E,\\\lambda_\ell = \lambda_{\ell+m}}} a_m e_m(y) \\
 + \frac{1}{T N_E} \sum_{m\neq 0}a_m e_m(y)\sum_{\substack{\ell\in\Z^d,\\ \lambda_\ell,\lambda_{\ell+m}\le E\\ \lambda_{\ell+m}\neq\lambda_\ell}}  \frac{\ee^{\ii T(\lambda_{\ell+m}-\lambda_\ell)}-1}{\ii(\lambda_{\ell+m}-\lambda_\ell)}
\\
 = \int_{\T_\ast^d} a(x)\,\dd x + \sum_{m\neq 0} a_m e_m(y) \cdot \frac{\# \{\ell: \lambda_{\ell}\le E,\lambda_{\ell+m}\le E, \lambda_{\ell+m}= \lambda_\ell\}}{N_E}\\
 + \frac{1}{T N_E} \sum_{m\neq 0}a_m e_m(y)\sum_{\substack{\ell\in\Z^d,\\ \lambda_\ell,\lambda_{\ell+m}\le E\\ \lambda_{\ell+m}\neq\lambda_\ell}}  \frac{\ee^{\ii T(\lambda_{\ell+m}-\lambda_\ell)}-1}{\ii(\lambda_{\ell+m}-\lambda_\ell)}\,.
\end{multline}

\textbf{The second term.}
Now $\lambda_k = 4\pi^2 k^2$, for $k^2:= k_1^2+\dots+k_d^2$. We have $N_E = \#\{\ell : \ell^2 \le \frac{E}{4\pi^2}\} \sim c_d E^{d/2}$ by known Weyl asymptotics, which say that $N_E$ is asymptotic to the volume of the $d$-dimensional ball of radius $\frac{\sqrt{E}}{2\pi}$. On the other hand, the constraint $\lambda_{\ell+m} = \lambda_\ell$ means that $\ell \cdot m = -m^2/2$. This defines an affine hyperplane in $\R^d$. Hence, for any $m\neq 0$, $\{\ell :\lambda_\ell \le E, \lambda_{\ell+m}=\lambda_\ell\}$ is itself the number of points on a $(d-1)$-dimensional ball of radius $\le \frac{\sqrt{E}}{2\pi}$, and as such, is bounded by $c_{d-1}E^{(d-1)/2}$, uniformly in $m\neq 0$ (by varying $m$ we may get fewer, but not more than $c_{d-1}E^{(d-1)/2}$ points). We thus see that
\[
\sup_{m\neq 0} \frac{\# \{\ell: \lambda_{\ell}\le E,\lambda_{\ell+m}\le E, \lambda_{\ell+m}= \lambda_\ell\}}{N_E} \to 0
\]
as $E\to\infty$. By \eqref{e:obsob}, it follows that the second error term in \eqref{e:conclucont} decays like $E^{-1/2}$.

\textbf{The third term.} Let us show that $\lim_{E\to\infty}\frac{1}{N_E}\sum\limits_{\substack{\ell\,:\,\lambda_\ell\le E\\ \lambda_{\ell+m}\neq\lambda_\ell}}  \frac{1}{|\lambda_{\ell+m}-\lambda_\ell|} = 0$ uniformly in $m$. Roughly speaking, this is a Ces\`aro argument (if $c_n\to 0$ then $\frac{1}{n}\sum_{k=1}^n c_k \to 0$).

Let $\varepsilon>0$ and fix $m\neq 0$, say $m_i\neq 0$ and write $\ell = (\hat{\ell}_i,\ell_i)$ with $\hat{\ell}_i\in \Z^{d-1}$. We have $\lambda_{\ell+m}-\lambda_\ell = 2\ell\cdot m+m^2$. So $\frac{1}{N_E}\sum_{\substack{\ell: \lambda_\ell \le E\\ |2\ell\cdot m+m^2|\ge \frac{2}{\varepsilon}}} \frac{1}{|\lambda_{\ell+m}-\lambda_\ell|}\le \frac{\varepsilon}{2}$.

On the other hand, if $B_E = \{\lambda_\ell\le E\}$, then
\begin{multline*}
\{\ell\in B_E: \ell \cdot m = 0\} = \Big\{\ell: \ell^2\le \frac{E}{4\pi} \text{ and } \ell_i = \frac{-\hat{\ell}_i\cdot \hat{m}_i}{m_i}\Big\} \\
= \Big\{\hat{\ell}_i: \hat{\ell}_i^2 + (\frac{-\hat{\ell}_i\cdot \hat{m}_i}{m_i})^2 \le \frac{E}{4\pi^2}\Big\} \subseteq \Big\{\hat{\ell}_i: \hat{\ell}_i^2\le \frac{E}{4\pi^2}\Big\},
\end{multline*}
so $|\{\ell\in B_E: \ell \cdot m = 0\}| \lesssim E^{(d-1)/2}$, with the implicit constant depending on the dimension, but not $m$. Similarly, note that $|2\ell\cdot m+m^2|<\frac{2}{\varepsilon}$ implies
\[
-\frac{2}{\varepsilon}-m^2-2\hat{\ell}_i\cdot \hat{m}_i<2 \ell_i m_i < \frac{2}{\varepsilon}-m^2-2\hat{\ell}_i\cdot \hat{m}_i. 
\]
There are at most $\frac{2}{|m_i|\varepsilon}+1$ values of $\ell_i$ in this interval. We see by applying the previous argument to each such value that $|\{\ell\in B_E : |2\ell\cdot m+m^2|<\frac{2}{\varepsilon}\}| \lesssim \varepsilon^{-1} E^{(d-1)/2}$, since $|m_i|\ge 1$. Summarizing, we have
\[
 \frac{1}{N_E}\sum\limits_{\substack{\ell\,:\,\lambda_\ell\le E\\ \lambda_{\ell+m}\neq\lambda_\ell}}  \frac{1}{|\lambda_{\ell+m}-\lambda_\ell|}\lesssim \frac{\varepsilon^{-1}}{\sqrt{E}} + \frac{\varepsilon}{2} \lesssim \frac{1}{E^{1/4}}
\]
by choosing $\varepsilon \asymp E^{-1/4}$. Using \eqref{e:obsob}, this implies the third term vanishes like $E^{-1/4}$.

This completes the proof of (1), and (3), as $\|a^E\|_1\le C_s\|a^E\|_{H^s}\le C_sE^r$ and $E^{r-\frac{1}{4}}\to 0$.

\textbf{Proof of (2).} $\langle \delta_x^E, \ee^{-\ii t\Delta} a \ee^{\ii t\Delta} \delta_y^E\rangle = \frac{1}{N_E} \sum\limits_{\substack{m,\ell\in \Z^d,\\ \lambda_\ell,\lambda_{m+\ell}\le E}} \ee^{2\pi\ii \ell\cdot(x-y)}a_m \ee^{\ii t(\lambda_{\ell+m}-\lambda_\ell)}e_m(y)$ by the same calculations. Let $x\neq y$. Note that $x_j-y_j\in (-1,1)$ for all $j$ since $x,y\in\T_\ast^d$, and $x_i-y_i\neq 0$ for at least one $i$. The term $m=0$ corresponds to $\frac{\langle a\rangle}{N_E}\sum_{\lambda_\ell\le E} \ee^{2\pi\ii \ell\cdot(x-y)}$, where $\langle a\rangle = \int_{\T_\ast^d} a(w)\dd w$.  Here $\lambda_\ell \le E \iff \ell^2 \le \frac{E}{4\pi^2}$. Consider for simplicity $d=1$, so the sum runs over $[-\frac{\sqrt{E}}{2\pi},\frac{\sqrt{E}}{2\pi}]$ and equals $\frac{\ee^{2\pi\ii \alpha (\beta+1)}-\ee^{-2\pi\ii \alpha\beta}}{\ee^{2\pi\ii \alpha}-1}$ for $\alpha=x-y$ and $\beta = \frac{\sqrt{E}}{2\pi}$. Since $\alpha\neq 0$, this may be bounded by some $c_\alpha$ independent of $E$, so $\frac{\langle a\rangle}{N_E}\sum_{\lambda_\ell\le E} \ee^{2\pi\ii \ell\cdot(x-y)} \to 0$.

In general, if $\alpha=x-y$, $\alpha_i\neq 0$ and we denote $\ell = (\ell_i,\hat{\ell}_i)$ with $\hat{\ell}_i\in \R^{d-1}$, then the finite sum over the ball $B_E = \{\ell^2\le \frac{E}{4\pi^2}\}$ can be rearranged into sections $\ell_i\in B_E(\hat{\ell}_i)$, for each $\hat{\ell}_i$ such that $(\ell_i,\hat{\ell}_i)\in B_E$. And each $\sum_{\ell_i\in B_E(\hat{\ell}_i)} \ee^{2\pi\ii \ell_i \alpha_i}$ can again be bounded by some $C_{\alpha_i}$ independently of $E$. By the Weyl asymptotics, we see that $|\sum_{\ell\in B_E} \ee^{2\pi\ii \ell\cdot \alpha}|\le C_{\alpha_i,d} E^{(d-1)/2}$. Since $N_E\sim E^{d/2}$, this implies $\frac{\langle a\rangle}{N_E}\sum_{\lambda_\ell\le E} \ee^{2\pi\ii \ell\cdot(x-y)} \to 0$.

We have shown that the term $m=0$ vanishes as $E\to\infty$. On the other hand, the sum over nonzero $m$ is controlled as in (1); the presence of the phase $\ee^{2\pi\ii \ell\cdot (x-y)}$ makes no difference. We conclude that if $x\neq y$, then $\frac{1}{T}\int_0^T \langle \ee^{\ii t\Delta}\delta_x^E,a\ee^{\ii t\Delta}\delta_y^E\rangle\,\dd t\to 0$ as $E\to\infty$.

\textbf{Weak convergence.} $\frac{1}{T}\int_0^T\langle \ee^{\ii t\Delta} \delta_y^E, a\ee^{\ii t\Delta}\delta_y^E\rangle\,\dd t = \frac{1}{T}\int_0^T \int_{\T_\ast^d}a(x)|(\ee^{\ii t\Delta}\delta_y^E)(x)|^2\,\dd x\dd t$. The continuous function $a$ is bounded on the compact $\T_\ast^d$, so $\int_0^T\int_{\T_\ast^d}|a(x)||(\ee^{\ii t\Delta}\delta_y^E)(x)|^2\,\dd x\dd t \le T\|a\|_\infty \|\ee^{\ii t\Delta}\delta_y^E\|^2 = T\|a\|_\infty$ is finite. By the Fubini theorem, we get $\frac{1}{T}\int_0^T\langle \ee^{\ii t\Delta} \delta_y^E, a\ee^{\ii t\Delta}\delta_y^E\rangle = \int_{\T_\ast^d} a(x)(\frac{1}{T}\int_0^T |(\ee^{\ii t\Delta}\delta_y^E)(x)|^2\,\dd t)\,\dd x = \int_{\T_\ast^d} a(x)\,\dd\mu_{y,T}^E(x)$.

It follows from (1) that $\int_{\T_\ast^d} a(x)\,\dd\mu_{y,T}^E(x) \to \int_{\T_\ast^d} a(x)\,\dd x$ for any Sobolev function $a$, in particular for any smooth function on $\T_\ast^d$. Combining \cite[Cor 15.3, Thm 13.34]{Klenke}, we deduce that $\dd\mu_{y,T}^E(x) \xrightarrow{w} \dd x$ as $E\to\infty$.
\end{proof} 

\begin{proof}[Proof of Lemma~\ref{lem:tav}]

Consider $d=1$ and $a(x) = e_1(x) = \ee^{2\pi\ii x}$. The calculation in the previous proof shows that
\begin{equation}\label{e:noave}
\langle \ee^{\ii t\Delta} \delta_y^E, a\ee^{\ii t\Delta}\delta_y^E\rangle = \frac{1}{N_E} \sum_{\ell:\lambda_\ell,\lambda_{\ell+1}\le E} \ee^{4\pi^2 \ii t(2\ell+1)} e_1(y)\,,
\end{equation}
where we used $a_m = \delta_{m,1}$ and $\lambda_k=4\pi^2 k^2$.

If $t=\frac{n}{4\pi}$, this gives $\frac{\ee^{\ii \pi n}e_1(y)}{N_E} \sum_{\lambda_\ell,\lambda_{\ell+1}\le E} \ee^{2\ell n \pi \ii} = c_E (-1)^n e_1(y)$ for $c_E = \frac{\sqrt{E}-\pi}{\sqrt{E}}$.

On the other hand, if $t = \frac{2n+1}{8\pi}$, then \eqref{e:noave} becomes $\frac{e_1(y)\ee^{\ii \frac{(2n+1)\pi}{2}}}{N_E}\sum_{\lambda_{\ell},\lambda_{\ell+1}\le E} \ee^{\ii (2n+1)\pi \ell} = \frac{e_1(y)\ee^{\ii \frac{(2n+1)\pi}{2}}}{N_E}\sum_{\lambda_{\ell},\lambda_{\ell+1}\le E} (-1)^\ell \in \{0, \pm \frac{e_1(y)\ee^{\ii \frac{(2n+1)\pi}{2}}}{N_E}\}$. 

The limits over $E$ are different: in the first case it gives $(-1)^n e_1(y)$, in the second case it gives $0$. The latter case corresponds to $\langle a\rangle=\langle e_1\rangle=0$, but not the former.
\end{proof}

\subsection{Generalizations}\label{sec:ctsgen}
It is not very clear what would be the analog of \eqref{e:ephipsi}. The limit $\lim_{T\to\infty} \frac{1}{T}\int_0^T \langle \ee^{\ii t\Delta}\phi, a\ee^{\ii t\Delta}\psi\rangle\,\dd t$ is not necessarily equal to $\langle \phi,\psi\rangle \langle a\rangle$ even if $\phi,\psi$ are smooth. For example, take $\phi=e_j$ and $\psi=e_k$. Then $\langle \ee^{\ii t\Delta}e_j, a\ee^{\ii t\Delta} e_k\rangle= \ee^{\ii t(\lambda_k-\lambda_j)}\langle e_j,ae_k\rangle$. We see that if $k\neq j$ but $\lambda_k=\lambda_j$ (e.g. $k=(0,1)$, $j=(1,0)$), then $\frac{1}{T}\int_0^T\langle \ee^{\ii t\Delta} e_j, a\ee^{\ii t\Delta} e_k\rangle = \langle e_j,ae_k\rangle$. Taking $a = e_{j-k}$, this has value $1$. In contrast, $\langle e_j,e_k\rangle\langle a\rangle = 0$.

Instead of $\delta_y^E$, we can consider variants such as $\chi_y^E := \frac{1}{\sqrt{\sum_j \chi_E(\lambda_j)^2}} \chi_E(-\Delta)\delta_y$. In other words, $\chi_y^E = \frac{1}{\sqrt{\sum_j \chi_E(\lambda_j)^2}} \sum_{j} \chi_E(\lambda_j) \overline{e_j(y)}e_j$. Here, instead of $\chi_E = \mathbf{1}_{\le E}$, we only ask $\chi_E(\lambda) = 0$ if $\lambda>E$ and $0<c_0\le \chi_E(\lambda)\le c_1$ on $[0,E-1]$. This allows for example to consider smooth cutoffs. Then the proof carries over. In fact, \eqref{e:conclucont} becomes
\begin{multline*}
\int_{\T_\ast^d} a(x)\,\dd x + \sum_{m\neq 0} a_m e_m(y) \cdot \frac{\sum_{\substack{\lambda_\ell\le E\\\lambda_\ell=\lambda_{\ell+m}}}\chi_E(\lambda_\ell)^2}{\sum_{\lambda_\ell\le E} \chi_E(\lambda_\ell)^2}\\
 + \frac{1}{T \sum_{\lambda_\ell\le E}\chi_E(\lambda_\ell)^2} \sum_{m\neq 0}a_m e_m(y)\sum_{\substack{\ell\in\Z^d,\\ \lambda_\ell,\lambda_{\ell+m}\le E\\ \lambda_{\ell+m}\neq\lambda_\ell}}  \frac{\ee^{\ii T(\lambda_{\ell+m}-\lambda_\ell)}-1}{\ii(\lambda_{\ell+m}-\lambda_\ell)}\chi_E(\lambda_\ell)\chi_E(\lambda_{\ell+m})\,.
\end{multline*}

For the second term, we bound the fraction by $\frac{c_1^2}{c_0^2}\cdot\frac{\#\{\ell:\lambda_\ell\le E,\,\lambda_{\ell+m}\le E,\,\lambda_{\ell+m}=\lambda_\ell\}}{N_{E-1}}$, which converges to zero uniformly in $m$ by the same argument. Similarly, the third term is controlled as before since $|\chi_E(\lambda_\ell)\chi_E(\lambda_{\ell+m})|\le c_1^2$ and $\sum_{\lambda_\ell\le E} \chi_E(\lambda_\ell)^2\ge c_0N_{E-1}$.

Finally, the proof can be generalized to tori of the form $\T = \times_{i=1}^d [0,b_i)$. Here we use the basis $e_\ell(x) = \frac{\ee^{\frac{2\pi x_1\ell_1}{b_1}}\cdots \ee^{\frac{2\pi x_d\ell_d}{b_d}} }{\sqrt{b_1\cdots b_d}}$, with eigenvalues $\lambda_\ell = 4\pi^2 \sum_{i=1}^d \frac{\ell_i^2}{b_i^2}$. The set $B_E = \{\lambda_\ell \le E\}$ now consists of points in an ellipsoid of axes $\frac{b_i\sqrt{E}}{2\pi}$. We still have $N_E\sim C_{b,d} E^{d/2}$ and the proof carries over mutatis mutandis. If we assume some irrationality condition, the second term decays faster with $E$ as the multiplicity reduces, however it seems that the third error term does not improve.

\subsection{Regularizing in position space}\label{sec:conpo}
Our aim here is to prove Theorem~\ref{thm:cont1}. We fix an arbitrary sequence $(\phi_\varepsilon)$ for $\varepsilon=(\varepsilon_1,\dots,\varepsilon_d)$ which satisfies the following:
\begin{itemize}
\item $\phi_\varepsilon = \mathop\otimes_{i=1}^d \phi_{\varepsilon_i}$, that is, $\phi_\varepsilon(x) = \phi_{\varepsilon_1}(x_1)\cdots \phi_{\varepsilon_d}(x_d)$ for some functions $\phi_{\varepsilon_i}$ on $\R$.
\item $\|\phi_{\varepsilon_i}\|=1$ for each $i$.
\item $\sup_{r\in\Z}|\langle \phi_{\varepsilon_i},e_r\rangle| \to 0$ as $\varepsilon_i\to 0$, where $e_r(s) = \ee^{2\pi\ii rs}$ for $s\in \T_\ast$.
\end{itemize}

The most important example is $\phi_\varepsilon = \frac{1}{\sqrt{\varepsilon_1\varepsilon_2\cdots\varepsilon_d}} \mathbf{1}_{\times_{i=1}^d [y_i,y_i+\varepsilon_i]}$. Here, $\langle \phi_{\varepsilon_i},e_r\rangle = \sqrt{\varepsilon_i}$ if $r=0$ and $\langle \phi_{\varepsilon_i},e_r\rangle = \frac{1}{\sqrt{\varepsilon_i}} \cdot \frac{\ee^{2\pi\ii r(y_i+\varepsilon_i)}-\ee^{2\pi\ii r y_i}}{2\pi r}$ if $r\neq 0$. Since $|\ee^{\ii x}-1|\le |x|$, we see that $|\langle \phi_{\varepsilon_i},e_r\rangle|\le \sqrt{\varepsilon_i}$. This can be regarded as a normalized point mass in the sense that if $\widetilde{\phi}_\varepsilon = \frac{1}{\varepsilon_1\cdots \varepsilon_d} \mathbf{1}_{\times_{i=1}^d [y_i,y_i+\varepsilon_i]}$, then for any integrable $g$, we have $\langle \widetilde{\phi}_\varepsilon, g\rangle \to g(y)$ for a.e. $y$ by the Lebesgue differentiation theorem. Also, $\widetilde{\phi}_\varepsilon(y) = \frac{1}{\varepsilon_1\cdots \varepsilon_d}\to \infty$ and $\widetilde{\phi}_\varepsilon(x)\to 0$ for $x\neq y$.

\begin{thm}\label{thm:ctspogen}
For any $(\phi_\varepsilon)$ and $(\psi_\varepsilon)$ as above, any $a\in H^s(\T_\ast^d)$, $s>d/2$, any $T>0$,
\[
\lim_{\varepsilon\downarrow 0} \frac{1}{T}\int_0^T \langle \ee^{\ii t\Delta} \phi_\varepsilon, a \ee^{\ii t\Delta} \psi_\varepsilon\rangle\,\dd t = \Big(\int_{\T_\ast^d} a(x)\,\dd x\Big)\Big(\lim_{\varepsilon\downarrow 0}\langle \phi_\varepsilon,\psi_\varepsilon\rangle\Big)\,,
\]
where $\varepsilon\downarrow 0$ means more precisely that $\varepsilon_i\downarrow 0$ for each $i$.
\end{thm}

If $\phi_\varepsilon=\psi_\varepsilon$, the scalar product on the right is $\|\phi_\varepsilon\|^2=1$. This implies Theorem~\ref{thm:cont1}. In fact, if we take $\psi_\varepsilon=\frac{1}{\sqrt{\varepsilon_1\varepsilon_2\cdots\varepsilon_d}} \mathbf{1}_{\times_{i=1}^d [y_i,y_i+\varepsilon_i]}$ and $\phi_\varepsilon = \frac{1}{\sqrt{\varepsilon_1\varepsilon_2\cdots\varepsilon_d}} \mathbf{1}_{\times_{i=1}^d [x_i,x_i+\varepsilon_i]}$ for $x\neq y$, then $\lim_{\varepsilon\downarrow 0}\langle \phi_\varepsilon,\psi_\varepsilon\rangle=0$.
\begin{proof}
We have $\ee^{\ii t\Delta} \psi = \sum_{\ell} \ee^{-\ii t\lambda_\ell}\psi_\ell e_\ell$ and $a=\sum_m a_m e_m$, so using $e_m e_{\ell} = e_{m+\ell}$, we get $\ee^{-\ii t\Delta} a\ee^{\ii t\Delta}\psi = \sum_{m,\ell} a_m \psi_\ell \ee^{\ii t(\lambda_{\ell+m}-\lambda_\ell)} e_{m+\ell}$. This implies that
\[
\langle \ee^{\ii t\Delta} \phi,a \ee^{\ii t\Delta}\psi\rangle = \sum_{m} a_m \sum_\ell \ee^{\ii t(\lambda_{\ell+m}-\lambda_\ell)} \psi_\ell \langle \phi,e_{\ell+m}\rangle.
\]
Thus, $\frac{1}{T}\int_0^T \langle \ee^{\ii t\Delta} \phi,a \ee^{\ii t\Delta}\psi\rangle\,\dd t = a_0 \sum_{\ell} \langle e_\ell,\psi\rangle\langle \phi,e_\ell\rangle + \sum_{m\neq 0} a_m \sum_{\ell:\lambda_{\ell}=\lambda_{\ell+m}}\langle e_\ell,\psi\rangle\langle \phi,e_{\ell+m}\rangle + \sum_{m\neq 0} a_m \sum_{\ell:\lambda_{\ell}\neq \lambda_{\ell+m}} \frac{\ee^{\ii T(\lambda_{\ell+m}-\lambda_\ell)}-1}{\ii T(\lambda_{\ell+m}-\lambda_\ell)}\langle e_\ell,\psi\rangle\langle \phi,e_{\ell+m}\rangle$.

For $\phi=\phi_\varepsilon$ and $\psi=\psi_\varepsilon$, the first term has the form
\[
a_0 \langle \phi_\varepsilon,\psi_\varepsilon\rangle = \langle \phi_\varepsilon,\psi_\varepsilon\rangle \int_{\T_\ast^d} a(x)\,\dd x\,.
\]

\textbf{Second term:} We prove that $\lim_{\varepsilon\to 0} \sum_{m\neq 0} a_m \sum_{\ell\,:\, 2\ell \cdot m = -m^2} \langle e_\ell, \psi_\varepsilon\rangle\langle \phi_\varepsilon,e_{\ell+m}\rangle = 0$.

First assume $d=1$. Then $m\neq 0$ and $2\ell m=-m^2$ implies $\ell = -m/2$. So we get $\sum_{m\neq 0} a_m \langle e_{-m/2},\psi_\varepsilon\rangle\langle \phi_\varepsilon,e_{m/2}\rangle$. By hypothesis, the general term vanishes as $\varepsilon\to 0$. Moreover, it is  bounded by $|a_m|\cdot \|\psi_\varepsilon\|\|\phi_\varepsilon\|\|e_{-m/2}\|\|e_{m/2}\| = |a_m|$, which is summable. By dominated convergence, the result follows.

Now let $d>1$. Using $|\sum_{m\neq 0} F(m)|\le \sum_{i=1}^d \sum_{m_i\neq 0} |F(m)|$, it suffices to show that $\lim_{\varepsilon\downarrow 0} \sum_{i=1}^d\sum_{m_i\neq 0} |a_m| |\sum_{\ell:2\ell\cdot m=-m^2}  \langle e_\ell, \psi_\varepsilon\rangle\langle \phi_\varepsilon,e_{\ell+m}\rangle| = 0$.

By hypothesis, $\phi_\varepsilon = \mathop\otimes_{i=1}^d \phi_{\varepsilon_i}$. It follows that $\langle e_k, \phi_\varepsilon \rangle = \prod_{j=1}^d \langle e_{k_j},\phi_{\varepsilon_j}\rangle$. Suppose that $m_i\neq 0$. Then $\ell_i = \frac{-m^2-2\hat{\ell}_i\cdot \hat{m}_i}{2 m_i}$, where for $x\in \R^d$, we denoted $x = (\hat{x}_i,x_i)$ with $\hat{x}_i\in \R^{d-1}$. We thus consider
\[
\sum_{i=1}^d \sum_{m_i\neq 0} |a_m| \sum_{\hat{\ell}_i\in \Z^{d-1}} |\langle e_{\ell_i},\psi_{\varepsilon_i}\rangle\langle \phi_{\varepsilon_i},e_{\ell_i+m_i}\rangle|\prod_{j\le d,j\neq i} |\langle e_{\ell_j},\psi_{\varepsilon_j}\rangle\langle \phi_{\varepsilon_j},e_{\ell_j+m_j}\rangle| \,.
\]
We first show the general term $F_\varepsilon(m)\to 0$ as $\varepsilon\to 0$. For this, we bound
\begin{multline*}
\sum_{\hat{\ell}_i\in \Z^{d-1}} |\langle e_{\ell_i},\psi_{\varepsilon_i}\rangle\langle \phi_{\varepsilon_i},e_{\ell_i+m_i}\rangle|\prod_{j\le d,j\neq i} |\langle e_{\ell_j},\psi_{\varepsilon_j}\rangle\langle \phi_{\varepsilon_j},e_{\ell_j+m_j}\rangle| \\
\le  \sup_{r\in \Z} |\langle e_r, \psi_{\varepsilon_i}\rangle|\sup_{k\in\Z}|\langle \phi_{\varepsilon_i},e_k\rangle| \sum_{\hat{\ell}_i\in \Z^{d-1}}\prod_{j\le d,j\neq i} |\langle e_{\ell_j},\psi_{\varepsilon_j}\rangle\langle \phi_{\varepsilon_j},e_{\ell_j+m_j}\rangle| \\
= \sup_{r\in \Z} |\langle e_r, \psi_{\varepsilon_i}\rangle|\sup_{k\in\Z}|\langle \phi_{\varepsilon_i},e_k\rangle| \prod_{j\le d,j\neq i} \sum_{\ell_j\in \Z}  |\langle e_{\ell_j},\psi_{\varepsilon_j}\rangle\langle \phi_{\varepsilon_j},e_{\ell_j+m_j}\rangle| \\
\le\sup_{r\in \Z} |\langle e_r, \psi_{\varepsilon_i}\rangle|\sup_{k\in\Z}|\langle \phi_{\varepsilon_i},e_k\rangle| \prod_{j\le d,j\neq i}\Big( \sum_{\ell_j\in \Z}  |\langle e_{\ell_j},\psi_{\varepsilon_j}\rangle|^2\Big)^{1/2}\Big(\sum_{\ell_j\in \Z}|\langle \phi_{\varepsilon_j},e_{\ell_j+m_j}\rangle|^2\Big)^{1/2} \\
= \sup_{r\in \Z} |\langle e_r, \psi_{\varepsilon_i}\rangle|\sup_{k\in\Z}|\langle \phi_{\varepsilon_i},e_k\rangle| \prod_{j\le d,j\neq i}\|\psi_{\varepsilon_j}\|\|\phi_{\varepsilon_j}\|\to 0
\end{multline*}
by hypothesis. On the other hand, the general term can be bounded by
\[
|a_m|\sum_{\ell:2\ell\cdot m=-m^2} |\langle e_{\ell},\psi_{\varepsilon}\rangle\langle e_{\ell+m},\phi_{\varepsilon}\rangle|\le |a_m| \Big(\sum_{\ell\in \Z^d} |\langle e_{\ell},\psi_{\varepsilon}\rangle|^2\Big)^{1/2}\Big(\sum_{\ell\in\Z^d}|\langle e_{\ell+m},\phi_{\varepsilon}\rangle|^2\Big)^{1/2}\le |a_m|
\]
which is summable. By dominated convergence, the result follows.

\textbf{Third term.} We need to show the following term vanishes as $\varepsilon\downarrow 0$:
\begin{equation}\label{e:contextraterm}
\sum_{m\neq 0} a_m \sum_{\ell:\lambda_{\ell+m}\neq \lambda_\ell} \frac{\ee^{\ii T(\lambda_{\ell+m}-\lambda_\ell)}-1}{\ii T(\lambda_{\ell+m}-\lambda_\ell)} \langle e_\ell,\psi_\varepsilon\rangle\langle \phi_\varepsilon, e_{\ell+m}\rangle\,.
\end{equation}

Let $m\neq 0$, say $m_i\neq 0$. We have
\[
\sum_{\ell\,:\, 2\ell\cdot m\neq -m^2} \frac{|\langle e_\ell,\psi_\varepsilon\rangle\langle \phi_\varepsilon, e_{\ell+m}\rangle|}{|2\ell\cdot m+m^2|} = \sum_{r\neq \frac{-m^2}{2}} \sum_{\ell\,:\, \ell\cdot m=r} \frac{|\langle e_\ell,\psi_\varepsilon\rangle\langle \phi_\varepsilon, e_{\ell+m}\rangle|}{|2r+m^2|} \,.
\]

Now $\ell \cdot m = r \iff \ell_i = \frac{r-\hat{\ell}_i\cdot \hat{m}_i}{m_i}$. Recalling $\phi_\varepsilon=\otimes\phi_{\varepsilon_i}$, this can be written as
\begin{multline*}
 \sum_{\hat{\ell}_i\in\Z^{d-1}} |\langle e_{\hat{\ell}_i},\psi_{\hat{\varepsilon}_i}\rangle\langle \phi_{\hat{\varepsilon}_i},e_{\hat{\ell}_i+\hat{m}_i}\rangle|\sum_{r\neq \frac{-m^2}{2}} \frac{|\langle e_{\frac{r-\hat{\ell}_i\cdot \hat{m}_i}{m_i}},\psi_{\varepsilon_i}\rangle\langle \phi_{\varepsilon_i},e_{\frac{m_i^2+r-\hat{\ell}_i\cdot \hat{m}_i}{m_i}}\rangle|}{|2r+m^2|}\\
 \le \sup_{k\in \Z} |\langle e_k,\psi_{\varepsilon_i}\rangle|\sum_{\hat{\ell}_i}  |\langle e_{\hat{\ell}_i},\psi_{\hat{\varepsilon}_i}\rangle\langle \phi_{\hat{\varepsilon}_i},e_{\hat{\ell}_i+\hat{m}_i}\rangle|\Big(\sum_{r\neq \frac{-m^2}{2}}\frac{1}{(2r+m^2)^2}\Big)^{1/2}\Big(\sum_{k\in \Z} |\langle \phi_{\varepsilon_i}, e_k\rangle|^2\Big)^{1/2}
\end{multline*}
where we used that $\frac{r-\hat{\ell}_i\cdot \hat{m}_i}{m_i}$ and $\frac{m_i^2+r-\hat{\ell}_i\cdot \hat{m}_i}{m_i}$ are both integers and the Cauchy-Schwarz inequality. 

If $m^2\in 2\Z$, then $ \frac{1}{4}\sum_{r\neq \frac{-m^2}{2}}\frac{1}{(r+\frac{m^2}{2})^2} = \frac{1}{4}\sum_{k\neq 0} \frac{1}{k^2} = \frac{\pi^2}{12}$.

If $m^2\notin 2\Z$, then $k_0:=\lfloor \frac{-m^2}{2}\rfloor = \frac{-m^2}{2}-\frac{1}{2}$. So $\frac{1}{4}\sum_{r}\frac{1}{(r+\frac{m^2}{2})^2} = \frac{1}{4}(\frac{1}{(1/2)^2} + \frac{1}{(1/2)^2} + \sum_{r\notin\{k_0,k_0+1\}} \frac{1}{(r-k_0-\frac{1}{2})^2})\le \frac{1}{4}(8+\sum_{k\neq 0}\frac{1}{k^2}) = 2+\frac{\pi^2}{12}$.

Either way, the sum is bounded by $3$. We thus showed that
\begin{multline*}
\sum_{\ell\,:\, 2\ell\cdot m\neq -m^2} \frac{|\langle e_\ell,\psi_\varepsilon\rangle\langle \phi_\varepsilon, e_{\ell+m}\rangle|}{|2\ell\cdot m+m^2|}\le \sqrt{3}\|\phi_{\varepsilon_i}\|\sup_{k\in \Z} |\langle e_k,\psi_{\varepsilon_i}\rangle| \sum_{\hat{\ell}_i\in\Z^{d-1}}  |\langle e_{\hat{\ell}_i},\psi_{\hat{\varepsilon}_i}\rangle\langle \phi_{\hat{\varepsilon}_i},e_{\hat{\ell}_i+\hat{m}_i}\rangle|\\
\le \sqrt{3}\|\phi_{\varepsilon_i}\| \|\psi_{\hat{\varepsilon}_i}\| \|\phi_{\hat{\varepsilon}_i}\| \sup_{k\in \Z} |\langle e_k,\psi_{\varepsilon_i}\rangle|
\end{multline*}
which tends to zero uniformly in $m$ by our hypotheses.

Finally turning back to \eqref{e:contextraterm}, we have proved that as a sum over $m$, the general term vanishes as $\varepsilon\downarrow 0$. Moreover, the general term is bounded by $|a_m| \frac{1}{T} \sum_{\ell} |\langle e_\ell,\psi_{\varepsilon}\rangle\langle \phi_\varepsilon,e_{\ell+m}\rangle| \le \frac{|a_m|}{T}$ by Cauchy-Schwarz and $\|\psi_\varepsilon\|=\|\phi_\varepsilon\|=1$. Since $\sum_m |a_m|<\infty$, we conclude by dominated convergence that \eqref{e:contextraterm} vanishes as $\varepsilon\downarrow 0$.
\end{proof}

\begin{rem}\label{rem:cts2lim}
It is clear that the proofs of Theorems~\ref{thm:ctstor} and \ref{thm:cont1} continue to hold if we take the limit $T\to\infty$ before considering $E\to\infty$. The proofs become in fact simpler as such a limit over $T$ kills the third term in \eqref{e:conclucont} and \eqref{e:contextraterm}, thereby avoiding the finer analysis we performed. See also \cite{KK} for this regime. If we do take the limit over $T\to\infty$ first, then we can also replace $\int_0^T$ by $\sum_{t=0}^{T-1}$. For example, \eqref{e:contextraterm} becomes $\sum_{m\neq 0} a_m \sum_{\ell:\lambda_{\ell+m}\neq \lambda_\ell} \frac{\ee^{\ii T(\lambda_{\ell+m}-\lambda_\ell)}-1}{\ii T(\ee^{\ii(\lambda_{\ell+m}-\lambda_\ell)}-1)} \langle e_\ell,\psi_\varepsilon\rangle\langle \phi_\varepsilon, e_{\ell+m}\rangle$, which vanishes as $T\to\infty$ (here $0\neq \lambda_{\ell+m}-\lambda_\ell = 4\pi^2(2\ell\cdot m+m^2)\notin 2\pi\Z$). We cannot however consider $\sum_{t=0}^{T-1}$ in Theorems~\ref{thm:ctstor} and \ref{thm:cont1}, as the finer analysis of the third term that we performed used the fact that $\frac{1}{|\lambda_{\ell+m}-\lambda_\ell|}\to 0$ as $\ell\to\infty$, which is not true of $\frac{1}{|\ee^{\ii(\lambda_{\ell+m}-\lambda_\ell)}-1|}$.
\end{rem}

\section{Case of the sphere}

Consider the sphere $\mathbb{S}^{d-1}\subset \R^d$, $d\ge 3$. We have $L^2(\mathbb{S}^{d-1}) = \mathop\oplus_{k=0}^\infty \mathbb{Y}_k^d$, where $\mathbb{Y}_k^d$ is the spherical harmonic space of order $k$ in dimension $d$. Any nonzero function in $\mathbb{Y}_k^d$ is an eigenfunction of the Laplacian on the sphere $-\Delta_{\mathbb{S}^{d-1}}$ with eigenvalue $k(k+d-2)$ and multiplicity $N_{k,d} = \dim \mathbb{Y}_k^d = \frac{(2k+d-2)(k+d-3)!}{k!(d-2)!}$. See \cite[Th. 2.38, Prp. 3.5]{AH}. If we define the \emph{zonal harmonic of degree $k$}, $Z_\xi^{(k)}(\eta) := \frac{N_{k,d}}{|\mathbb{S}^{d-1}|} P_{k,d}(\xi\cdot \eta)$, where $|\mathbb{S}^{d-1}| = \frac{2\pi^{d/2}}{\Gamma(d/2)}$ is the volume of the sphere and $P_{k,d}(t)$ satisfies the Poisson identity $\sum_{k=0}^\infty r^k N_{k,d}P_{k,d}(t) = \frac{1-r^2}{(1+r^2-2rt)^{d/2}}$ for $|r|<1$ and $t\in [-1,1]$, then $Z_\xi^{(k)}$ satisfies the reproducing property \cite[(2.33)]{AH},
\begin{equation}\label{e:repro}
\psi(\xi) = \langle Z_\xi^{(k)}, \psi\rangle_{L^2(\mathbb{S}^{d-1})} \qquad \forall \psi\in \mathbb{Y}_k^d\,, \quad \xi\in \mathbb{S}^{d-1}\,.
\end{equation}

We have $Z_\xi^{(k)} = \sum_{j=1}^{N_{k,d}} Y_{k,j}(\xi) \overline{Y_{k,j}}$ for any orthonormal basis of $\mathbb{Y}_k^d$, see \cite[Th. 2.9]{AH}. In particular $Z_\xi^{(k)}\in \mathbb{Y}_k^d$ and as such $Z_\xi^{(k)}$ is an eigenfunction of $-\Delta_{\mathbb{S}^{d-1}}$ for the eigenvalue $k(k+d-2)$. Moreover, $\|Z_\xi^{(k)}\|_{L^2(\mathbb{S}^{d-1})}^2 = \frac{N_{k,d}}{|\mathbb{S}^{d-1}|}$ by \cite[(2.40)]{AH} and $Z_\xi^{(k)}(\xi) = \frac{N_{k,d}}{|\mathbb{S}^{d-1}|}$ by \cite[(2.35)]{AH}, which is the maximum of $Z_\xi^{(k)}$. Finally, if $C_{n,\nu}$ is the Gegenbauer ultraspherical polynomial, then $C_{n,\frac{d-2}{2}}(t) = \binom{n+d-3}{n}P_{n,d}(t)$ for $d\ge 3$, \cite[(2.145)]{AH}.

In view of the Dirac-like equation \eqref{e:repro}, one can consider the evolution of
\begin{equation}\label{e:1stsug}
\langle \ee^{\ii t\Delta} \widetilde{Z}_\xi^{(k)}, a \ee^{\ii t\Delta} \widetilde{Z}_\xi^{(k)}\rangle\,,
\end{equation}
where $\widetilde{Z}_\xi^{(k)} = \sqrt{\frac{|\mathbb{S}^{d-1}|}{N_{k,d}}} Z_\xi^{(k)}$ has norm one. 

However, since $Z_\xi^{(k)}$ is an eigenfunction of $-\Delta_{\mathbb{S}^{d-1}}$, this reduces to $\langle \widetilde{Z}_\xi^{(k)},a \widetilde{Z}_\xi^{(k)}\rangle$. 

\begin{lem}\label{lem:predens}
The density $|\widetilde{Z}_\xi^{(k)}(\eta)|^2$ is not uniformly distributed as $k\to\infty$. It has peaks at $\pm \xi$ and stays bounded for $\eta \neq \pm\xi$.
\end{lem}
\begin{proof}
In fact, $|\widetilde{Z}_\xi^{(k)}(\pm \xi)|^2 = \frac{|\mathbb{S}^{d-1}|}{N_{k,d}}\cdot \frac{N_{k,d}^2}{|\mathbb{S}^{d-1}|^2} \to \infty$ as $k\to\infty$.  On the other hand, if $\eta\neq \pm \xi$, then $t:= \xi \cdot \eta$ satisfies $|t|<1$. By \cite[(2.117)]{AH}, $|P_{n,d}(t)| < \frac{\Gamma(\frac{d-1}{2})}{\sqrt{\pi}}[\frac{4}{n(1-t^2)}]^{\frac{d-2}{2}}$, hence $|\widetilde{Z}_\xi^{(k)}(\eta)|^2<\frac{N_{k,d}}{|\mathbb{S}^{d-1}|}\cdot\frac{c_{d,\eta}}{k^{d-2}}$. Since $\Gamma(x+\alpha)\sim  \Gamma(x)x^\alpha$, then $N_{k,d} =\frac{(2k+d-2)}{(d-2)!}\frac{\Gamma(k+d-2)}{\Gamma(k+1)} \sim \frac{2}{(d-2)!}k^{d-2}$, hence $|\widetilde{Z}_\xi^{(k)}(\eta)|^2$ stays bounded as $k\to\infty$. The upper bound we used is sharp in $n$, i.e. $P_{n,d}(t)\asymp \frac{1}{n^{\frac{d-2}{2}}}$, by \cite[Th. 8.21.8]{Sz}.
\end{proof}

A closer analogue to $\delta_y^E = \frac{1}{\sqrt{N_E}}\sum_{\lambda_j\le E} \overline{e_j(y)}e_j$, our Dirac truncation for the torus, would be $S_\xi^{(n)} = \frac{1}{\sqrt{M_{n,d}}}\sum_{k=0}^n \mu_{n,k,d} Z_\xi^{(k)}$, where $M_{n,d}=\sum_{k=0}^n \mu_{n,k,d}^2 \frac{N_{k,d}}{|\mathbb{S}^{d-1}|} $ and $\mu_{n,k,d} = \frac{n!(n+d-2)!}{(n-k)!(n+k+d-2)!}$. In fact, if $R_\xi^{(n)} = \sum_{k=0}^n \mu_{n,k,d} Z_{\xi}^{(k)}$, then (see \cite[\S 2.8.1]{AH}):
\begin{itemize}
\item $\langle R_\xi^{(n)},f\rangle_{L^2(\mathbb{S}^{d-1})} \to f(\xi)$ uniformly in $\xi$, for any continuous $f$,
\item $R_\xi^{(n)}(\xi) = \sum_{k=0}^n \mu_{n,k,d}\frac{N_{k,d}}{|\mathbb{S}^{d-1}|} = E_{n,d} =\frac{(n+d-2)!}{(4\pi)^{\frac{d-1}{2}}\Gamma(n+\frac{d-1}{2})} \to \infty$,
\item $\|R_\xi^{(n)}\|_{L^2(\mathbb{S}^{d-1})}^2 = \|\sum_{k=0}^n \mu_{n,k,d}Z_\xi^{(k)}\|_2^2 = \sum_{k=0}^n \mu_{n,k,d}^2 \|Z_\xi^{(k)}\|^2  = M_{n,d}$.

Here we used that $Z_\xi^{(k)}\perp Z_\xi^{(r)}$ for $k\neq r$, as they belong to $\mathbb{Y}_k^d$ and $\mathbb{Y}_r^d$, respectively, which are in direct sum (they are in distinct eigenspaces).
\end{itemize}
This says that $R_\xi^{(n)}$ is a $\delta_\xi$-sequence and that $S_\xi^{(n)}$ is its normalization.

\begin{lem}\label{lem:dendi}
The density $p_\xi^{(n)}(\eta) = \lim_{T\to\infty} \frac{1}{T}\int_0^T |\ee^{\ii t\Delta}S_\xi^{(n)}(\eta)|^2\,\dd t$ is not uniformly distributed as $n\to\infty$. It has peaks at $\pm \xi$ and stays bounded for $\eta \neq \pm\xi$.
\end{lem}
\begin{proof}
As $\ee^{\ii t\Delta} S_\xi^{(n)} = \frac{1}{\sqrt{M_{n,d}}} \sum_{k=0}^n \ee^{-\ii t \lambda_k}\mu_{n,k,d} Z_\xi^{(k)}$, $\lambda_k = k(k+d-2)$, then $|\ee^{\ii t\Delta} S_\xi^{(n)}(\eta)|^2 = \frac{1}{M_{n,d}}\sum_{k=0}^n \mu_{n,k,d}^2|Z_\xi^{(k)}(\eta)|^2 + \frac{1}{M_{n,d}} \sum_{k\neq k'} \ee^{\ii t(\lambda_{k'}-\lambda_k)}\mu_{n,k,d}\mu_{n,k',d} Z_\xi^{(k)}(\eta)\overline{Z_\xi^{(k')}(\eta)}$. But $\lambda_k\neq\lambda_{k'}$ for $k\neq k'$, because the eigenspaces $\mathbb{Y}_k^d$ are in direct sum (alternatively, $\lambda_k=\lambda_{j}$ for $k\neq j$ would imply $(k^2-j^2)+(d-2)(k-j)=0$, so $k+j+d-2=0$, a contradiction since $d\ge 3$ and $k,j\ge 0$). It follows that $\frac{1}{T}\int_0^T |\ee^{\ii t\Delta} S_\xi^{(n)}(\eta)|^2\,\dd t \to \frac{1}{M_{n,d}} \sum_{k=0}^n \mu_{n,k,d}^2|Z_\xi^{(k)}(\eta)|^2=:p_\xi^{(n)}(\eta)$. 

Now $p_\xi^{(n)}(\pm\xi) = \frac{1}{M_{n,d}}\sum_{k=0}^n \mu_{n,k,d}^2\frac{N_{k,d}^2}{|\mathbb{S}^{d-1}|^2}$ diverges as $n\to\infty$. In fact, $(\sum_{k=0}^n a_k)^2 \le (n+1)\sum_{k=0}^n a_k^2$. Applying this to $a_k = \mu_{n,k,d}\frac{N_{k,d}}{|\mathbb{S}^{d-1}|}$, we get $E_{n,d}^2\le (n+1) \sum_{k=0}^n a_k^2$, so $\frac{1}{M_{n,d}}\sum_{k=0}^n a_k^2\ge \frac{E_{n,d}^2}{(n+1)M_{n,d}}  \to \infty$ for $d>3$, because $\frac{E_{n,d}^2}{(n+1)M_{n,d}} \ge \frac{E_{n,d}}{n+1} \asymp \frac{n^{\frac{d-1}{2}}}{n} \to \infty$. Here we used that $\mu_{n,k,d}\le 1$. For $d=3$, this lower bound is not useful. So we argue as follows.

We have $\sum_{k=0}^n \mu_{n,k,d}\frac{N_{k,d}}{|\mathbb{S}^{d-1}|} = E_{n,d}$. By Cauchy-Schwarz, we have on the one hand $(\sum_{k=0}^n \mu_{n,k,d}\frac{N_{k,d}}{|\mathbb{S}^{d-1}|})^2 \le (n+1) \sum_{k=0}^n \mu_{n,k,d}^2 \frac{N_{k,d}^2}{|\mathbb{S}^{d-1}|^2}$, so that $\sum_{k=0}^n  \mu_{n,k,d}^2 \frac{N_{k,d}^2}{|\mathbb{S}^{d-1}|^2} \ge \frac{E_{n,d}^2}{n+1}$. On the other hand, $M_{n,d}^2=(\sum_{k=0}^n \mu_{n,k,d}^2 \frac{N_{k,d}}{|\mathbb{S}^{d-1}|})^2\le (\sum_{k=0}^n \mu_{n,k,d}^2)(\sum_{k=0}^n \mu_{n,k,d}^2\frac{N_{k,d}^2}{|\mathbb{S}^{d-1}|^2})$ and so $\frac{\sum_{k=0}^n \mu_{n,k,d}^2 \frac{N_{k,d}^2}{|\mathbb{S}^{d-1}|^2}}{M_{n,d}} =(\sum_{k=0}^n \mu_{n,k,d}^2 \frac{N_{k,d}^2}{|\mathbb{S}^{d-1}|^2})^{1/2}\cdot \frac{(\sum_{k=0}^n \mu_{n,k,d}^2 \frac{N_{k,d}^2}{|\mathbb{S}^{d-1}|^2})^{1/2}}{M_{n,d}} \ge \frac{(\sum_{k=0}^n \mu_{n,k,d}^2 \frac{N_{k,d}^2}{|\mathbb{S}^{d-1}|^2})^{1/2}}{(\sum_{k=0}^n \mu_{n,k,d}^2)^{1/2}}$. The two inequalities imply that $\frac{\sum_{k=0}^n \mu_{n,k,d}^2 \frac{N_{k,d}^2}{|\mathbb{S}^{d-1}|^2}}{M_{n,d}} \ge \frac{E_{n,d}}{\sqrt{n+1}(\sum_{k=0}^n \mu_{n,k,d}^2)^{1/2}}$.

Specializing to $d=3$, we have $E_{n,3} = \frac{n+1}{4\pi}$. Also, $\sum_{k=0}^n \mu_{n,k,3}^2 = \sum_{k=0}^n (\frac{n!(n+1)!}{(n-k)!(n+k+1)!})^2 = \sum_{k=0}^\infty \frac{(1)_k(-n)_k(-n)_k}{(n+2)_k(n+2)_kk!} = {}_3 F_2(1,-n,-n;n+2,n+2;1)$. Here $(r)_k = r(r+1)\cdots(r+k-1)$. If $r\in\N^\ast$, then $(r)_k = \frac{(r+k-1)!}{(r-1)!}$ for any $k$, while $(-r)_k = -r(-r+1)\cdots (-r+k-1) = (-1)^k r(r-1)\cdots (r-k+1) = (-1)^k \frac{r!}{(r-k)!}$ holds for $k\le r$ and $(-r)_k=0$ for $k>r$.

By Dixon's identity, ${}_3 F_2(1,-n,-n;n+2,n+2;1) = \frac{\Gamma(\frac{3}{2})\Gamma(\frac{3}{2}+2n)\Gamma(2+n)^2}{\Gamma(2)\Gamma(2+2n)\Gamma(\frac{3}{2}+n)^2}$. Using the asymptotic $\Gamma(x+\alpha)\sim \Gamma(x)x^\alpha$, this shows that $\sum_{k=0}^n \mu_{n,k,3}^2 \sim \frac{\sqrt{\pi}}{2}\cdot \frac{(n^{1/2})^2}{(2n)^{1/2}} = \frac{1}{2}\sqrt{\frac{\pi n}{2}}$.

We thus get that $\frac{\sum_{k=0}^n \mu_{n,k,3}^2 \frac{N_{k,3}^2}{|\mathbb{S}^{2}|^2}}{M_{n,3}} \gtrsim\frac{n+1}{4\pi\sqrt{n+1}(\frac{1}{2}\sqrt{\frac{\pi n}{2}})^{1/2}} \asymp n^{1/4}\to\infty$.

We showed that $p_\xi^{(n)}(\pm \xi)\to \infty$. If $\eta\neq \pm \xi$, then as in Lemma~\ref{lem:predens}, we have $p_\xi^{(n)}(\eta)< \frac{1}{M_{n,d}} \sum_{k=0}^n \mu_{n,k,d}^2 \frac{N_{k,d}^2}{|\mathbb{S}^{d-1}|^2} \frac{c_{d,\eta}}{k^{d-2}}$ with $N_{k,d}  \sim \frac{2}{(d-2)!}k^{d-2}$, so that $M_{n,d} = \sum_{k=0}^n \mu_{n,k,d}^2 \frac{N_{k,d}}{|\mathbb{S}^{d-1}|}$ and $\sum_{k=0}^n \mu_{n,k,d}^2 \frac{N_{k,d}^2}{|\mathbb{S}^{d-1}|^2} \frac{c_{d,\eta}}{k^{d-2}}$ grow at the same speed with $n$, and $p_\xi^{(n)}(\eta)$ stays bounded.
\end{proof}

To see more explicitly that ergodicity is violated, we can compare averages of specific observables. For simplicity, let us choose $a$ such that for the $\xi$ we fixed, $a(\eta) = a(\xi\cdot \eta)$. This allows to use the Funk-Hecke formula \cite[Th. 2.22]{AH}: for any $Y_n\in\mathbb{Y}_n^d$,
\[
\int_{\mathbb{S}^{d-1}} f(\xi\cdot \eta)Y_n(\eta)\,\dd S^{d-1}(\eta)=Y_n(\xi)|\mathbb{S}^{d-2}|\int_{-1}^1 P_{n,d}(t)f(t)(1-t^2)^{\frac{d-3}{2}}\,\dd t
\]
For $Y_n = Z_\xi^{(n)}$ and $f(\xi \cdot \eta) = a(\xi \cdot \eta) P_{n,d}(\xi\cdot \eta)\frac{N_{n,d}}{|\mathbb{S}^{d-1}|}=a(\xi\cdot \eta)Z_{\xi}^{(n)}(\eta)$ we get
\begin{align*}
\int_{\mathbb{S}^{d-1}} a(\xi\cdot \eta)[Z_\xi^{(n)}(\eta)]^2\,\dd S^{d-1}(\eta)&=\frac{Z_\xi^{(n)}(\xi)|\mathbb{S}^{d-2}| N_{n,d}}{|\mathbb{S}^{d-1}|}\int_{-1}^1 [P_{n,d}(t)]^2a(t)(1-t^2)^{\frac{d-3}{2}}\,\dd t\\
&=\frac{|\mathbb{S}^{d-2}| N_{n,d}^2}{|\mathbb{S}^{d-1}|^2}\int_{-1}^1 [P_{n,d}(t)]^2a(t)(1-t^2)^{\frac{d-3}{2}}\,\dd t \,.
\end{align*}
On the other hand (case $n=0$ of Funk-Hecke, cf. \cite[(2.87)]{AH}),
\begin{equation}\label{e:n=0fh}
\int_{S^{d-1}} a(\xi\cdot\eta)\,\dd S^{d-1}(\eta) = |\mathbb{S}^{d-2}|\int_{-1}^1 a(t)(1-t^2)^{\frac{d-3}{2}}\,\dd t \,.
\end{equation}

Let us check the asymptotics of \eqref{e:1stsug}. We take $d=3$. Then $|\mathbb{S}^{d-1}|=4\pi$, $|\mathbb{S}^{d-2}|=2\pi$, $N_{k,d} = 2k+1$ and $P_{n,d}(t) = P_n(t)$ is the Legendre polynomial, which satisfies \cite[(2.79)]{AH}
\[
\int_{-1}^1 P_n(t)P_m(t)\,\dd t = \frac{2}{2n+1}\delta_{n,m} \,.
\]
Since $(n+1)P_{n+1}(t)+nP_{n-1}(t)=(2n+1)tP_n(t)$, see \cite[p. 52]{AH}, we deduce that
\[
\int_{-1}^1 t^2 P_n(t)^2\,\dd t = \frac{2}{(2n+1)^2}\Big(\frac{n^2}{2n-1}+\frac{(n+1)^2}{2n+3}\Big)\,.
\]

\begin{lem}
Let $d=3$. There exists an observable $a$ such that for any $T>0$,
\[
\lim_{k\to\infty} \frac{1}{T}\int_0^T\langle \ee^{\ii t\Delta}\widetilde{Z}_\xi^{(k)},a\ee^{\ii t\Delta}\widetilde{Z}_\xi^{(k)}\rangle\,\dd t =\lim_{k\to\infty} \langle \widetilde{Z}_\xi^{(k)},a\widetilde{Z}_\xi^{(k)}\rangle \neq \langle a\rangle\,.
\]
\end{lem}
\begin{proof}
Given $\xi$, we have for $a(\eta) := a(\xi \cdot \eta)$, with $a(t)=t^2$, that
\begin{align*}
\langle \widetilde{Z}_\xi^{(k)},a \widetilde{Z}_\xi^{(k)}\rangle &= \frac{|\mathbb{S}^{d-1}|}{N_{k,d}}\int_{\mathbb{S}^{d-1}} a(\xi\cdot \eta)[Z_\xi^{(k)}(\eta)]^2\,\dd S^{d-1}(\eta) \\
&= \frac{|\mathbb{S}^{d-2}|N_{k,d}}{|\mathbb{S}^{d-1}|}\int_{-1}^1 [P_{k,d}(t)]^2a(t)(1-t^2)^{\frac{d-3}{2}}\,\dd t\\
&=\frac{(2\pi)(2k+1)}{4\pi}\int_{-1}^1 t^2 P_k(t)^2\,\dd t= \frac{1}{2k+1}\Big(\frac{k^2}{2k-1}+\frac{(k+1)^2}{2k+3}\Big)\,.
\end{align*}

This tends to $\frac{1}{2}$ as $k\to\infty$. On the other hand, by \eqref{e:n=0fh},
\begin{equation}\label{e:avasph}
\langle a\rangle = \frac{1}{|\mathbb{S}^{d-1}|}\int_{\mathbb{S}^{d-1}} a(\xi\cdot\eta)\,\dd S^{d-1}(\eta) = \frac{1}{2}\int_{-1}^1 t^2\,\dd t = \frac{1}{3}\,.
\end{equation}

This completes the proof.
\end{proof}

A similar phenomenon holds for $S_\xi^{(n)}$. 

\begin{lem}
Let $d=3$. There exists an observable $a$ such that
\[
\liminf\limits_{n\to\infty} \lim\limits_{T\to\infty} \frac{1}{T}\int_0^T \langle \ee^{\ii t\Delta} S_\xi^{(n)},a\ee^{\ii t\Delta}S_\xi^{(n)}\rangle\,\dd t  \neq \langle a\rangle \,.
\]
\end{lem}
\begin{proof}
Our arguments in Lemma~\ref{lem:dendi} show that
\[
\lim_{T\to\infty} \frac{1}{T}\int_0^T \langle \ee^{\ii t\Delta} S_\xi^{(n)},a\ee^{\ii t\Delta}S_\xi^{(n)}\rangle\,\dd t = \frac{1}{M_{n,d}} \sum_{k=0}^n \mu_{n,k,d}^2\int_{\mathbb{S}^{d-1}} a(\eta) |Z_\xi^{(k)}(\eta)|^2\,\dd S^{d-1}(\eta)\,.
\]
Choosing again $d=3$ and $a(\eta) = a(\xi\cdot \eta)$ for $a(t) = t^2$, this becomes
\[
\frac{1}{M_{n,3}} \sum_{k=0}^n\mu_{n,k,3}^2 \frac{(2\pi)(2k+1)^2}{(4\pi)^2}\int_{-1}^1 t^2 P_k(t)^2\,\dd t = \frac{1}{M_{n,3}} \sum_{k=0}^n \mu_{n,k,3}^2 \frac{1}{4\pi}\Big(\frac{k^2}{2k-1}+\frac{(k+1)^2}{2k+3}\Big).
\]
But the expression in parentheses is
\[
\frac{2k^3+3k^2+(2k-1)(k^2+2k+1)}{(2k-1)(2k+3)} = \frac{(2k+1)(2k^2+2k-1)}{4k^2+4k-3}
\]
and $\frac{2k^2+2k-1}{4k^2+4k-3} \ge \frac{1}{2}$ for $k\ge 1$. Thus, as $\mu_{n,0,3}=1$, the limit is
\begin{multline*}
\ge \frac{1}{12\pi M_{n,3}} + \frac{1}{2M_{n,3}} \sum_{k=1}^n \mu_{n,k,3}^2\frac{2k+1}{4\pi} = \frac{1}{12\pi M_{n,3}} + \frac{1}{2M_{n,3}}\sum_{k=0}^n \mu_{n,k,3}^2\frac{N_{k,3}}{|\mathbb{S}^2|} - \frac{1}{8\pi M_{n,3}}\\
=\frac{-1}{24\pi M_{n,3}} + \frac{1}{2} \to \frac{1}{2}\,,
\end{multline*}
since $M_{n,3} = \sum_{k=0}^n \mu_{n,k,3}^2\frac{N_{k,3}}{|\mathbb{S}^2|}\ge \frac{1}{4\pi}\sum_{k=0}^n \mu_{n,k,3}^2 \asymp \sqrt{n}\to\infty$ as we saw in Lemma~\ref{lem:dendi}. Summarizing, we have shown that
\[
\liminf_{n\to\infty} \lim_{T\to\infty} \frac{1}{T}\int_0^T \langle \ee^{\ii t\Delta} S_\xi^{(n)},a\ee^{\ii t\Delta}S_\xi^{(n)}\rangle\,\dd t \ge \frac{1}{2} \,.
\]
In view of \eqref{e:avasph}, this completes the proof.
\end{proof}


\medskip

The results of this section suggest that unitary evolution of the Laplacian on the sphere does not make point masses equidistributed as time goes on, hence a lack of ergodicity. This is in accord with the classical picture. 

Still, our analysis is based on the evolution of some specific normalized $\delta$-sequence $S_\xi^{(n)}$. It would be interesting to see if equidistribution continues to be violated for other choices.

\appendix
\section{Discussion}\label{app:a}
\subsection{Eigenfunction thermalization}\label{sec:et}
One could heuristically deduce the present dynamical criterion of ergodicity using quantum ergodicity of eigenvectors as follows. Suppose that $H_N$ is a self-adjoint operator on a finite graph $G_N$ of order $N$ having a quantum ergodic basis $(\psi_j^{(N)})$. Fix a normalized initial state $\phi$, assume for simplicity that it is well-defined as $N$ varies (e.g. $G_N\subset G_{N+1}$ or $G_{N+1}$ is a cover of $G_N$). Expanding $\phi = \sum_{k=1}^N \langle \psi_k^{(N)},\phi\rangle \psi_k^{(N)}$, so that $\ee^{-\ii t H_N} \phi = \sum_{k=1}^N \langle \psi_k^{(N)},\phi\rangle \ee^{-\ii t\lambda_k^{(N)}}\psi_k^{(N)}$, we obtain
\begin{multline*}
\langle \ee^{-\ii t H_N}\phi, a_N \ee^{-\ii tH_N}\phi\rangle = \Big\langle \sum_{m=1}^N \langle \psi_m^{(N)},\phi\rangle\ee^{-\ii t\lambda_m^{(N)}}\psi_m^{(N)}, a_N\sum_{n=1}^N \langle \psi_n^{(N)},\phi\rangle \ee^{-\ii t\lambda_n^{(N)}}\psi_n^{(N)}\Big\rangle\\
=\sum_{n=1}^N |\langle \psi_n^{(N)},\phi\rangle|^2 \langle \psi_n^{(N)}, a_N\psi_n^{(N)}\rangle + \sum_{\substack{m,n\le N\\ m\neq n}} \ee^{\ii t(\lambda_m^{(N)}-\lambda_n^{(N)})} \overline{\langle \psi_m^{(N)},\phi\rangle}\langle \psi_n^{(N)},\phi\rangle\langle \psi_m^{(N)},a_N\psi_n^{(N)}\rangle
\end{multline*}
From here one could argue that $\langle \psi_m^{(N)}, a\psi_n^{(N)}\rangle \approx \delta_{m,n}\langle a\rangle$ if the eigenfunctions satisfy a strong form of quantum ergodicity. If not, one could take a time average $\frac{1}{T}\int_0^T$ and assume the spectrum is simple, so that the double sum of oscillatory terms vanishes as $T\to\infty$. Since $\sum_{n=1}^N |\langle \psi_n^{(N)},\phi\rangle|^2 \langle \psi_n^{(N)}, a_N\psi_n^{(N)}\rangle \approx \sum_{n=1}^N |\langle \psi_n^{(N)},\phi\rangle|^2\langle a_N\rangle = \langle a_N\rangle \|\phi\|^2=\langle a_N\rangle$ by the assumed eigenfunction ergodicity, we get that $\langle \ee^{-\ii t H_N}\phi, a_N \ee^{-\ii tH_N}\phi\rangle \approx \langle a_N\rangle$ or $\lim_{T\to\infty}\frac{1}{T}\int_0^T \langle \ee^{-\ii t H_N}\phi, a_N \ee^{-\ii tH_N}\phi\rangle\,\dd t \approx \langle a_N\rangle$, respectively. The same heuristic can be used in the continuum, for example on the torus.

This heuristics is folklore in the physics community and is commonly known as \emph{eigenfunction thermalization}. It can be made rigorous for some models of random matrices, see for example the discussion in \cite{CES}. Note that it only requires one special basis of eigenfunctions $(\psi_j^{(N)})$ to be ergodic in a strong sense, e.g. quantum uniquely ergodic. For specific (deterministic) graphs it seems to us that making this argument rigorous can be more difficult than proving from scratch. In particular, for the case of $\Z^d$, the eigenvalues have a high multiplicity, which complicates this scheme even though we have a very nice ergodic basis $\psi_j^{(N)}(n) = \frac{1}{N^{d/2}} \ee^{2\pi\ii j\cdot n/N}$ at disposal. Also note that on the torus, this conclusion is simply wrong if we choose $\phi(x)=\sqrt{2}\cos(2\pi x)$ since $\ee^{\ii t\Delta}\phi = \ee^{-4\pi^2 \ii t}\phi$, so at least in this model where the spectrum is highly degenerate, there must be some assumption on $\phi$.

\subsection{Other interpretations} It may be interesting to explore other quantum dynamical interpretations of ergodicity. 

Since in the classical picture on the torus, we calculate the mean value $\frac{1}{T}\int_0^T a(x_0+ty_0)\,\dd t$ of an observable $a$ over an orbit of $x_0$, one could naively replace $a(y)$ by $\langle \delta_y, a\rangle$ and thus consider the limit of $\frac{1}{T}\int_0^T \langle \ee^{\ii t\Delta}\delta_y,a\rangle\,\dd t$ and see if it approaches $\int a(x)\,\dd x$. Let us check what this gives.

Let $H$ be a self-adjoint operator on a Hilbert space $\mathscr{H}$ and let $\phi,\psi\in\mathscr{H}$. Then by the functional calculus we have
\begin{equation}\label{e:fc}
\lim_{T\to\infty} \frac{1}{T}\int_0^T \langle \phi,\ee^{\ii tH}\psi\rangle\,\dd t= \langle \phi,\chi_{\{0\}}(H)\psi\rangle\,.
\end{equation}


This interpretation would rely entirely on the spectral projection at $0$. It does not seem entirely convincing. For example, both the torus and sphere Laplacian satisfy that $\chi_{\{0\}}(-\Delta)$ is the orthogonal projection onto the flat function, that is, $(\chi_{\{0\}}(-\Delta_{\T_\ast^d})a)(y) = \int_{\T_\ast^d} a(x)\,\dd x$ and $(\chi_{\{0\}}(-\Delta_{\mathbb{S}^{d-1}})a)(y) = \frac{1}{|\mathbb{S}^{d-1}|}\int_{\mathbb{S}^{d-1}} a(x)\,\dd S^{d-1}(x)$, so that the RHS of \eqref{e:fc} in both cases is $\langle \phi\rangle\langle \psi\rangle$, although the dynamics are very different in each case.
%

\subsection{Further comments} 
In general, ergodic theorems apply to integrable functions $f$ and assert that $\lim_{n\to\infty}\frac{1}{n}\sum_{k=1}^n f(T^k x) = g(x)$, where $g$ is such that $\int g = \int f$. If $T$ is ergodic, then $g$ must be constant, and this constant is $\frac{1}{\mu(\Omega)} \int_{\Omega} f(y)\,\dd\mu(y)$.

On the quantum side, the long-time convergence of $\frac{1}{T}\int_{0}^T |(\ee^{\ii tH}\psi)(x)|^2$ is quite trivial if $\mathscr{H}$ is finite dimensional (e.g. working on a finite graph). In fact, if $P_k$ are the spectral projections for the distinct eigenvalues of $H$, then
\[
|(\ee^{\ii tH}\psi)(x)|^2 = \Big|\sum_{k=1}^m \ee^{\ii tE_k} (P_k\psi)(x)\Big|^2 = \sum_{k=1}^m |(P_k\psi)(x)|^2 + \sum_{k\neq \ell} \ee^{\ii t(E_k-E_\ell)}(P_k\psi)(x)\overline{(P_\ell\psi)(x)}
\]

Since $\frac{1}{T}\int_0^T \ee^{\ii t(E_k-E_\ell)}\,\dd t = \frac{1}{\ii T}\cdot \frac{\ee^{\ii T(E_k-E_\ell)}-1}{E_k-E_\ell} \to 0$, we get that
\begin{equation}\label{e:clasergo}
\lim_{T\to\infty} \frac{1}{T}\int_0^{T} |(\ee^{\ii tH}\psi)(x)|^2\,\dd t = \sum_{k=1}^m |(P_k\psi)(x)|^2\,.
\end{equation}
The function on the RHS clearly satisfies that its integral (or its sum over the graph) is equal to $\psi$. So the hard part here is to prove ergodicity rather than convergence, i.e. proving that for certain models the RHS is constant. This is what we did for the case of cubes, asymptotically in $N$. We do not use \eqref{e:clasergo} for this purpose, as the high multiplicity of eigenvalues complicates the calculation of $P_k\psi$. We emphasize that for the torus, we do not need to consider long time $T$, so \eqref{e:clasergo} is useless.

In the context of quantum walks, the papers \cite{Kar,CLR} investigate the mixing time, namely, how fast does $\frac{1}{T}\int_0^{T} |(\ee^{\ii tH}\psi)(x)|^2\,\dd t$ becomes $\epsilon$-close to its limit $\sum_{k=1}^m |(P_k\psi)(x)|^2$.

\subsection*{Acknowledgments}
We thank Maxime Ingremeau and Fabricio Maci\`a for explanations on the literature, Daniel Lenz for a question which we now answer in Proposition~\ref{prp:noavzd}, St\'ephane Nonnenmacher for interesting feedback that lead to some generalizations, and Ivan Veseli\'c for a question that lead to Proposition~\ref{prp:ves}.

\providecommand{\bysame}{\leavevmode\hbox to3em{\hrulefill}\thinspace}
\providecommand{\MR}{\relax\ifhmode\unskip\space\fi MR }
\providecommand{\MRhref}[2]{%
}
\providecommand{\href}[2]{#2}

\end{document}